\documentclass[journal]{IEEEtran}
\IEEEoverridecommandlockouts
\usepackage[noadjust]{cite}
\usepackage{amsmath,amssymb,amsfonts}
\usepackage{amsthm}
\usepackage{algorithmic}
\usepackage{relsize}
\usepackage{algorithm}
\usepackage{graphicx}
\usepackage{subcaption}
\usepackage{comment}
\usepackage{textcomp}
\usepackage{xcolor}
\usepackage{url}
\usepackage{xcolor}
\usepackage{tikz}
\def\BibTeX{{\rm B\kern-.05em{\sc i\kern-.025em b}\kern-.08em
		T\kern-.1667em\lower.7ex\hbox{E}\kern-.125emX}}
\usepackage{enumerate}
\usepackage{paralist} 
\usepackage{afterpage}
\usepackage{amsbsy}

\usepackage{mathtools}
\mathtoolsset{showonlyrefs}
\usepackage{placeins}

\allowdisplaybreaks
\DeclareMathOperator{\arctantwo}{arctan2}
\begin{document}
	
		\title{Robust Target Localization in 2D: A Value-at-Risk Approach}
	\author{Jo\~ao Domingos and Jo\~ao Xavier
	\thanks{\hspace{-.61cm} Jo\~ao Domingos and Jo\~ao Xavier {\tt\small oliveira.domingos@ist.utl.pt, jxavier@isr.tecnico.ulisboa.pt.}
		are with the Instituto Sistemas e Rob\'otica  and Instituto Superior T\'ecnico in Portugal. This work has been supported by FCT, Fundaç\~ao para a Ci\^encia e a Tecnologia,
		under the projects PD/BD/150631/2020 and LARSyS - FCT Project UIDB/50009/2020.}
}
%
\newtheorem{theorem}{Theorem}
\newtheorem{lemma}{Lemma}
\newtheorem{remark}{Remark}
\newtheorem{example}{Example}
\newtheorem{result}{Result}
\newtheorem{definition}{Definition}
\newtheorem{corollary}{Corollary}[theorem]
\maketitle

\begin{abstract}
	This paper consider considers the problem of locating a two dimensional target from range-measurements containing outliers. Assuming that the number of outlier is known, we formulate the problem of minimizing inlier losses while ignoring outliers. This leads to a combinatorial, non-convex, non-smooth problem involving the percentile function. Using the framework of risk analysis from Rockafellar et al., we start by interpreting this formulation as a Value-at-risk (VaR) problem from portfolio optimization. To the best of our knowledge, this is the first time that a localization problem was formulated using risk analysis theory. To study the VaR formulation, we start by designing a majorizer set that contains any solution  of a general percentile problem. This set is useful because, when applied to a localization scenario in 2D, it allows to majorize the solution set in terms of singletons, circumferences, ellipses and hyperbolas. Using know parametrization of these curves, we propose a grid method for the original non-convex problem. So we reduce the task of optimizing the VaR objective  to that of efficiently sampling the proposed majorizer set.  We compare our algorithm with four benchmarks in target localization. Numerical simulations show that our method is fast while, on average, improving the accuracy of the best benchmarks by at least $100$m in a $1$ Km$^2$ area.
\end{abstract}

	\maketitle
	
\section{Introduction}
\label{sec:introduction}
This paper considers the problem of robust target estimation, that is, estimating the position of a target $x\in \mathbf{R}^2$ from range measurements, some of which may be outliers. In concrete, we consider the additive model
\begin{align}
y_m=\left\|x-a_m\right\|+u_m, \enspace m=1,\dots,M
\label{eqn:data_model}
\end{align} 
with measurements $y_m\in \mathbf{R}$, anchors $a_m \in \mathbf{R}^2$ and additive uncertainties $u_m\in \mathbf{R}$. Several papers consider estimation problems where the uncertainty vector $u:=(u_1,\dots,u_m)$ lies in a known bounded region  which reflects the support of inlier error distributions~\cite{eldar2008minimax},~\cite{shi2016robust},~\cite{domingos2022robust}. Let us note, however, that in practice model~\eqref{eqn:data_model} can be affected by outlier uncertainties capable of biasing non-robust estimates by alarming amounts (see the numerics of section~\ref{sec:numerical_results}). In this paper we partition the set of $M$ measurements into inliers $\mathcal{I}$ and outliers $\mathcal{O}$, i.e.,
\begin{align}
\{1,\dots,M\}=\mathcal{I} \cup \mathcal{O},\enspace \mathcal{I} \cap \mathcal{O}=\emptyset.
\end{align}
In broad terms, measurement $y_m$ is an outlier if it deviates largely from the assumed data model $\left\|x-a_m\right\|$, that is, the difference $ y_{m} -\left\|x-a_m\right\| $ is large in absolute value. In this paper we assume that the number of outliers measurements $L\in\{0,\dots,M-1\}$ is known.
\textcolor{black}{ Let $\{\pi_0,\dots,\pi_{M-1}\}$ denote the permutation of $\{1,\dots,M\}$} that orders the absolute deviations $| y_{m} -\left\|x-a_m\right\| |$ in descending order, that is,
\textcolor{black}{ \begin{align}
| y_{\pi_0} -\left\|x-a_{\pi_0}\right\| | \geq \dots \geq | y_{\pi_{M-1}} -\left\|x-a_{\pi_{M-1}}\right\| |.
\end{align}}
Given $L$, the outlier set $\mathcal{O}$ corresponds to the set of measurements which have the $L$ largest absolute deviations so \textcolor{black}{ $\mathcal{O}=\{\pi_0,\dots,\pi_{L}\}$} with $\mathcal{O}=\emptyset$ if $L=0$. Note that, in general, the permutation \textcolor{black}{ $\{\pi_0,\dots,\pi_{M-1}\}$ }depends on target $x$ and, hence, is unknown. For simplicity of notation we ignore the dependency \textcolor{black}{ of indices  $\{\pi_0,\dots,\pi_{M-1}\}$ } on $x$. From now on we assume that only the number of outliers $L$ is know;  \textcolor{black}{not the indices $\{\pi_0,\dots,\pi_{M-1}\}$ of outlier measurements.}
\subsection{Problem Formulation}
We consider the problem of estimating the target position $x$ by knowing, in advance, that $L$ measurements deviate from model~\eqref{eqn:data_model}. In concrete, we want an estimate $\hat{x}$ such that the model deviations $\{| y_{m} -\left\|\hat{x}-a_m\right\| |\}_{m=1}^M$ are as low as possible for the inlier measurements $m\in \mathcal{I}$ while discarding/ignoring outlier measurements $m\in \mathcal{O}$. We compute $\hat{x}$ by minimizing the percentile objective
\begin{align}
\hat{x} \in \arg \min_x  p_{L} \{ | y_{{m}} -\left\|x-a_{m}\right\| | \},
\label{eqn:percentile_min}
\end{align}
with $p_{L}:\mathbf{R}^M\mapsto \mathbf{R}$ the $L$-th percentile function. The percentile function $p_{L}(z)$ returns the largest element of $z\in \mathbf{R}^M$ after discarding its $L$ largest entries. So, $p_{0}(z)=\max(z_i)$ and $p_{M-1}(z)=\min(z_i)$. In compact notation,  \textcolor{black}{$p_{L}(z)=z_{p_{L}}$ with $\{\pi_0,\dots\pi_{M-1}\}$ the permutation } of $\{1,\dots,M\}$ that orders the elements of $z$ in descending order, that is
\textcolor{black}{ \begin{align}
z_{\pi_0}\geq \dots \geq z_{\pi_{L}} = p_L(z)\geq z_{\pi_{L+1}} \geq \dots \geq z_{\pi_{M-1}}.
\label{eqn:sort_elements}
\end{align}}
In general, problem~\eqref{eqn:percentile_min} is difficult to solve because both the percentile function $p_L$ and the deviation mapping $x\mapsto | y_{m} -\left\|x-a_{m}\right\| |$ are non-convex and non-differentiable.
\subsection{Literature Review}
\label{sec:related-wk}
\paragraph{Target Localization}
Most literature on target localization considers a least squares (LS) approach that aims to minimize the sum of residuals between the measurement model and fixed range measurements~\cite{Teboulle_ADMM_2017,Majorization_minimization_2019,Target_Tracking_by_GD_2021,Source_chapter_2018,Ite_Source_Localization_2007,Approximate_ML_in_2D,Taylor_series_algorithm_2019,Optimal_Trilateration_2019}. The popularity of this formulation stems from its connections with maximum likelihood estimation in the presence of Gaussian noise~\cite{Becks_paper_2008}. Although statistical significant, the least squares formulation is non-convex and, hence, difficult to solve exactly in general. Over the past twenty years there have been several ideas to approach this inherently challenging problem. From our perspective, most of these ideas cluster into three algorithmic families: (1) classical algorithms like gradient descent~\cite{Target_Tracking_by_GD_2021,Ite_Source_Localization_2007,Reweighted_LS_2018} and Newton method~\cite{Source_chapter_2018,Clustering_with_TOA_Intersec_2018}; (2) semidefinite (SDP) relaxations that exploit the quadratic nature of the problem~\cite{Approximate_ML_in_2D,Taylor_series_algorithm_2019,Exact_by_Lagrange_2014,Balacing_pameter_2019,NLOS_Environments_2015,domingos2022robust}; (3) trust region approaches that reformulate the problem as a quadratically constrained quadratic program~\cite{Becks_paper_2008,Bisection_for_exact_2017,Differentiable_R_LS_2017,Reweighted_LS_2018}. On a more theoretical front, we highlight the recent work of Pun et al.~\cite{Target_Tracking_by_GD_2021} which shows that the LS objective, although non-convex, is locally strongly convex at its global minima. This property is relevant for first order methods since it enables global convergence for  ``good enough''  initializations. We latter compare our approach with the gradient method of~\cite{Target_Tracking_by_GD_2021}. 
~\\
\indent  It is well known, however, that traditional least squares formulations are highly sensitive to outlier measurements~\cite{STRONG_2021,Bootstrapping_2004,ADMMforsparse2021,Balacing_pameter_2019,NLOS_Environments_2015,survey_robust,Reweighted_LS_2018,ADMMforsparse2021}. In localization, outliers can come from non-line-of-sight (NLOS) propagation conditions, typically due to indoor or dense environments~\cite{NLOS_Environments_2015,Balacing_pameter_2019}. Given this limitation, there have been several approaches towards the robustification of the original problem. Earlier work by Sun et al.~\cite{Bootstrapping_2004} focuses on a bootstrapping scheme and Huber M-estimation. Most recently, Soares et al.~\cite{STRONG_2021} also consider Huber estimation for soft rejection of outliers. The employed Huber loss is non-convex but the authors derive tight convex underestimators that lead to tight convex relaxations. The work of Zaeemzadeh et al.~\cite{Reweighted_LS_2018} achieves robustness by considering Geman-Mclcure updates to re-fit the observed measurements in a iterative fashion. The problem formulation is again routed in M-estimation and, in each iteration, the employed algorithm uses the trust region results of~\cite{Becks_paper_2008} to update the position of the target. We compare our method with both M-estimates~\cite{STRONG_2021,Reweighted_LS_2018}. A different viewpoint is to account for NLOS biases explicitly in the original LS formulation. In ~\cite{NLOS_Environments_2015} Vaghefi et al. propose a semidefinite relaxation that simultaneously estimate the targets position and the NLOS biases of the model. More recently, Wang et al.~\cite{Balacing_pameter_2019} consider a LS objective for worst case NLOS biases, assuming known error bounds. The problem is relaxed into a convex SDP by means of the S-lemma.
\paragraph{Value-at-Risk optimization} The proposed percentile formulation is well-known in portfolio optimization. In concrete, problem~\eqref{eqn:percentile_min} is known as a  value-at-risk (VaR) problem in the context of  risk analysis~\cite{CVAR_1999}. This connection is formalized latter in section~\ref{sec:risk_analysis} after we review some background on this field. Here, we describe the algorithmic literature for this optimization class, which is rooted in economic theory~\cite{Quantil_Regression_2001,Regression_Quantiles_1978,CVAR_1999}. From our perspective, there have been four main ideas to approach VaR problems like the one in~\eqref{eqn:percentile_min}: (1) conditional value-at-risk (CVaR) methods based on the seminal work of Rockafellar et al.~\cite{CVAR_1999}; (2) difference-of-convex (DC) approaches that decompose the percentile function as a difference of two convex functions~\cite{DC_opt_+_branch_and_bound_2010,DC_opt_diff_convex_algorithm}; (3) integer programming schemes that explore the combinatorial nature of the percentile objective~\cite{Opt_return_w_VaR_const_2007,Near_opt_VaR_portfolios_2015,Practical_algorithms_VaR_2015,MIP_formulations_2017} and (4) smoothing techniques that filter out local, erratic modes of the VaR objective~\cite{Smooth_VaR,Smooth_Prob}. The CVaR approach is, perhaps, the most popular since, for convex losses, it leads to a convex problem that upper bounds the VaR objective~\cite{CVAR_1999}. So decisions with a low CVaR will implicitly also have a low VaR. In the context of portfolio optimization the resulting CVaR problem is a linear program~\cite{Algorithms_for_VaR_2000}, which can be solved efficiently by standard convex methods~\cite{boyd2004convex}. Furthermore, linearity actually enables efficient iterative schemes that try to improve the original VaR objective, per iteration~\cite{Algorithms_for_VaR_2000,CVaR_proxies_2014}. Let us note, however, that we found no method on the VaR literature that can be used to approach problem~\eqref{eqn:percentile_min}. The main issue is that, in our case, the percentile function is composed with non-convex maps $x\mapsto \big|y_m-||x-a_m||\big|$. All methods described so far assume that the percentile function is composed with linear~\cite{Opt_return_w_VaR_const_2007,Near_opt_VaR_portfolios_2015,Practical_algorithms_VaR_2015,Algorithms_for_VaR_2000,DC_opt_diff_convex_algorithm} or general convex maps~\cite{CVAR_1999}.
\subsection{Contributions}
\label{sec:Contributions}  
We claim three main contributions:
\begin{enumerate}
	\item Theoretical Analysis: We present a novel majorization inequality for an optimization class that we denote as selection problems (see Section~\ref{subsection:selection_problems}). This optimization class is rich enough to include percentile problems of the form~\eqref{eqn:percentile_min}. Our analysis (Theorem~\ref{thm:majorization_theorem}) produces a first order majorizer on this challenging optimization class which, when applied to formulation~\eqref{eqn:percentile_min}, actually leads to a simple, yet highly effective localization algorithm -- RTPE;
	\item Statistical interpretation: We interpret problem~\eqref{eqn:percentile_min} under the framework of risk analysis~\cite{CVAR_1999}. We show that formulation~\eqref{eqn:percentile_min} is actually minimizing the Value-At-Risk (VaR) measure for a risk problem with outliers ( Section~\ref{sec:risk_analysis});
	\item Numerical validation: We validate our approach by considering a setup with ten anchors ($M=10$) in a square area of $1$ Km$^2$. We compare algorithm~\ref{alg:RPTE} with four state-of-the art benchmarks in target localization~\cite{Becks_paper_2008},~\cite{Target_Tracking_by_GD_2021},~\cite{Reweighted_LS_2018},~\cite{STRONG_2021}. ~
	 When we have a reasonable number of outlier measurements (say $L=3,4,5$) and their deviation from model~\eqref{eqn:data_model} is moderate to high ($1$ Km to $2.5$ Km) our method improves the accuracy of the best benchmark by at least $100$m in a$ 1$ Km$^2$area. For low model deviations ($500$ m to $750$ m), our approach still improves the best benchmark but with lower accuracy gains ( $\approx 10$ meters).
\end{enumerate}
\subsection{Paper Organization}
\label{sec:organization}  
The remaining of the paper is organized as follows. Section~\ref{sec:risk_analysis} shows that problem~\eqref{eqn:percentile_min} admits a natural interpretation as minimizing the value-at-risk (VaR) risk measure from portfolio optimization. Section~\eqref{subsection:selection_problems} proves a majorization bound for an optimization class that includes the percentile problem~\eqref{eqn:percentile_min}. This majorization bound is applied to our problem in section~\ref{subsection: Parametrizing}; in this case the derived bound is computationally tractable in the sense that it corresponds to the union of easily parametrized regions in 2D space. Section~\ref{subsection:unbounded_Phi} deals with the unboundedness of some of these regions, in order to construct a fully implementable algorithm for problem~\eqref{eqn:percentile_min} -- see Algorithm~\ref{alg:RPTE}. Section~\ref{sec:numerical_results} provides numerical evidence that, on average, our method tends to outperforms four benchmarks in target localization. Section~\ref{sec:conclusion} concludes the paper.
 \section{Robust target estimation as a Risk Analysis problem}
 \label{sec:risk_analysis}
 The percentile function $p_{L}$ is of primary importance in risk analysis~\cite{CVAR_1999}. For completeness, we review the theoretical setup of risk analysis~\cite{CVAR_1999} that introduces  the Value-at-Risk (VaR) risk measure. Afterwards we show how the robust estimation scheme of~\eqref{eqn:percentile_min} is actually optimizing the Value-at-Risk (VaR) measure for an underlying stochastic risk problem.
 \subsection{Value-At-Risk}
 Generically, risk analysis considers the problem of decision making under uncertainty and different risk measures. We are given a loss function $f(x,Z)\in \mathbf{R}$ which depends on a decision vector $x\in \mathbf{R}^{n_x}$ and on a random vector $Z\in \mathbf{R}^{n_Z}$. The primary goal of risk analysis is the design of a decision vector $x$ such that the loss function $f(x,Z)$ is \textit{typically} low. We use the keyword \textit{typically} because, for each decision vector $x$, the loss function $f(x,Z)$ is a random variable hence the mapping $x\mapsto f(x,Z)$ is non deterministic.
 \textcolor{black}{One standard risk measure~\cite{larsen2002algorithms} is the so called Value-at-Risk (VaR). In simple terms, the $\beta-$VaR of a decision vector $x$ is the maximum loss $f(x,Z)$ that can be incurred with probability at least $\beta$. In concrete $\beta-$VaR, denoted as $u_\beta(x)$, is defined as
 \begin{align}
 u_\beta(x):=\inf\{u \in \mathbf{R}: \mathbb{P} (f(x,Z)\leq u) \geq \beta  \}.
 \label{eqn:VAR_definition}
 \end{align}
 Ideally we would like to find the decision vector $x$ that minimizes $u_\beta(x)$ for a fixed confidence level $\beta\in[0,1]$ so
 \begin{align}
 \hat{x} \in \arg \min_x u_\beta(x).
 \label{eqn:VAR_problem}
 \end{align}
 In general, problem~\eqref{eqn:VAR_problem} cannot be exactly solved either because the distribution of the random vector $Z$ may be unknown or because the objective $u_\beta(x)$ may be computationally intractable. A common fix~\cite{Distributional_Interpretation}  is to collect  independent and identically distributed (i.i.d.) samples $\{z_1,\dots,z_M\}$ of $Z$ and consider the empirical distribution 
 \begin{align}
 u\mapsto \frac{1}{M} \sum_{m=1}^M \mathbf{1}_{\{f(x,z_m)\leq u\}}
 	\label{eqn:empirical_dist_def}
 \end{align}
  with $(x,u)\mapsto \mathbf{1}_{\{f(x,z_m)\leq u\}}$ the indicator function of the event $\{f(x,z_m)\leq u\}$. So $\mathbf{1}_{\{f(x,z_m)\leq u\}}=1$ if variables $x$ and $u$ are such that $f(x,z_m)\leq u$ and $\mathbf{1}_{\{f(x,z_m)\leq u\}}=0$ otherwise. By using the empirical distribution~\eqref{eqn:empirical_dist_def} on definition~\eqref{eqn:VAR_definition} we get the approximate $\beta-$VaR
 \begin{align}
 	x\mapsto \inf\big\{u: \frac{1}{M} \sum_{m=1}^M \mathbf{1}_{\{f(x,z_m)\leq u\}} \geq \beta  \big\}.
	\label{eqn:VAR_problem_approximate}
\end{align}
Assume now, for simplicity, that $\beta$ is a multiple of of $1/M$. In this case the infimum in~\eqref{eqn:VAR_problem_approximate} is exactly the $(1-\beta)\, M$ percentile of the observations $f(x,z_m)$, that is
 \begin{align}
	\inf\big\{u: \frac{1}{M} \sum_{m=1}^M \mathbf{1}_{\{f(x,z_m)\leq u\}} \geq \beta  \big\}=p_{(1-\beta)\, M} \{ f(x,z_m)  \}
	\label{eqn:VAR_problem_approximate_2}
\end{align}
Minimizing the approximate objective~\eqref{eqn:VAR_problem_approximate_2} leads to
 \begin{align}
\min_x p_{(1-\beta)\, M} \{f(x,z_m)  \}.
\label{eqn:VAR_problem_smapled}
\end{align}}
So, as seen, a sampled VaR problem is equivalent to a percentile formulation where we optimize the worst case loss after discarding the $(1-\beta)\,M$ highest losses $f(x,z_m)$.
 \subsection{Risk Analysis Interpretation of the Robust Target Estimator}
Considering measurements~\eqref{eqn:data_model}, we start by modelling both the uncertainties $u_m$ and anchor positions $a_m$ as realization  of two random objects: an additive uncertainty $U$ in $\mathbf{R}$ and a anchor position vector $A$  in $\mathbf{R}^2$. Here we assume that the anchor positions $a_m$ are observed realizations of $A$, while the additive uncertainties  $u_m$ are unobserved samples of $U$. Given $U$ and $A$ we interpret~\eqref{eqn:data_model} as a sampled version of 
 \begin{align}
 Y=||x-A||+U.
 \label{eqn:model_probability}
 \end{align} 
  Let $Z=(Y, A)$ denote the random vector of observed quantities. Both these quantities are observable since, in our model, we estimate the position of target $x$ from measurements $y_m$ of $Y$ and anchor positions $a_m$ of $A$.
 Given $Z$, we estimate $x$ by considering the deviation loss between the true model $||A-x||$ and measurements $Y$
  \begin{align}
 f(x;Z):=| \left\|A-x\right\|-Y | ,\enspace Z=(Y, A).
 \label{eqn:target_loss_function}
 \end{align}
  So we are considering a risk analysis problem  where the measurements $Y$ and anchors $A$ are observed random objects with an underlying joint distribution and we want to design/estimate a target position $x$ such that model~\eqref{eqn:model_probability} is accurate, that is, the additive uncertainty $U$ is small according to~\eqref{eqn:target_loss_function}. 
  ~\\
  As explained in section~\ref{sec:introduction}, in a robust estimation problem the measurements $Y$ tend to be affect, most of the time, by some form of inlier noise but, sporadically, there exist outlier measurements 
 in model~\eqref{eqn:model_probability}. We model this behaviour by assuming that, in $M$ total measurements, $L$ of them are outlier. Given $L$, it is reasonable to choose $\beta$, the confidence level of the risk analysis problem, equal to $1-L/M$, the ratio  of inlier measurements.  
 So, in a risk analysis context, we can compute the decision vector $x$ by minimizing the $\beta-$VaR risk measure with confidence level $\beta=1-L/M$. This boils down to solving problem~\eqref{eqn:VAR_problem} with loss~\eqref{eqn:target_loss_function}. Given samples $(y_m,a_m)$ of $(Y,Z)$ we use approximation~\eqref{eqn:VAR_problem_smapled} to get
 \begin{align}
 \hat{x} \in \arg \min_x p_{L  } \{ | y_{{m}} -\left\|x-a_{m}\right\| | \}.
 \label{eqn:VAR_problem_localization_sample}
 \end{align}
 So the problem of estimating the target position $x$, through~\eqref{eqn:percentile_min}, actually corresponds to minimizing the empirical VaR risk measure, from the observed data pairs $(a_m,y_m)$.
\section{A Simple Majorizing Algorithm}
\label{sec:majorizing_algorithm}
In this section we detail our algorithm for solving problem~\eqref{eqn:percentile_min}. Our method is simple and consists in majorizing the solution set of problem~\eqref{eqn:percentile_min}, that is, we design a set $\Phi\subseteq \mathbf{R}^2$ that contains any minimizer of~\eqref{eqn:percentile_min}, so
\begin{align}
\arg \min_x  p_{L} \{ |y_m- \left\|x-a_m\right\|   |   \} \subseteq \Phi.
\label{eqn:majorization_ineq}
\end{align}
Our set $\Phi$ is computationally tractable, being the union of easily parameterized regions in 2D space. In concrete, set $\Phi$ is comprised of singletons, circumferences, ellipses and bounded half-hyperbolas. Given that all these regions can be efficiently parameterized our algorithm is simple: we create a grid over $\Phi$ and choose the best grid point for the objective~\eqref{eqn:percentile_min} -- see algorithm~\ref{alg:RPTE} in section~\ref{subsection:unbounded_Phi}. Our empirical results show that a fine grid is able to outperform several benchmarks in robust localization, while sharing similar computing times.
\subsection{Designing the Majorizing  Set $\Phi$ }
\label{subsection:selection_problems}
 In this section we construct a non-trivial set $\Phi$ such that~\eqref{eqn:majorization_ineq} holds. In fact our technique holds for a broader class of problems that we denote as selection problems. Given a dimension $D$, let $\{f_m\}_{m=1}^M$ denote a collection of real-valued functions on $\mathbf{R}^D$ and
  $\{S_m\}_{m=1}^M\subseteq2^{\mathbf{R}^D}$ a partition\footnote{So the sets $S_m$ are disjoint while covering $\mathbf{R}^D$, i.e., $\mathbf{R}^D=S_1 \cup \dots \cup S_M$.} of $\mathbf{R}^D$. Now let $f$ denote the function that is equal to ${f}_m$ on $S_m$ so
\begin{align}
f(x):= \mathbf{1}_{S_1}(x) {f}_1(x)+\dots+\mathbf{1}_{S_M}(x) {f}_M(x)
\label{eqn:selection_function}
\end{align}
In words, function $f$ selects among $M$ possible atoms  $\{{f}_m\}_{m=1}^M$ and the selection rule is described by sets $\{S_m\}_{m=1}^M$ such that $f(x)=f_m(x)$ whenever $x\in S_m$. A function of the form~\eqref{eqn:selection_function} is denoted as an $M$-th selector.
\begin{definition}
	\label{def:selection_function}
	A function $f:\mathbf{R}^D \mapsto \mathbf{R}$ is an $M$-th selector if there exists functions $\{f_m\}_{m=1}^M$ and a partition $\{S_m\}_{m=1}^M$ of $\mathbf{R}^D$ such that $f$ selects ${f}_m$ on $S_m$, that is
\begin{align}
f(x)= \mathbf{1}_{S_1}(x) f_1(x)+\dots+\mathbf{1}_{S_M}(x) f_M(x).
\label{eqn:selection_function_2}
\end{align}
Let  $\mathcal{S}_M$ denote the collection of $M$-th order selectors.
\end{definition}
We consider the minimization of an $M$-th selector $f\in \mathcal{S}_M$,
\begin{align}
\min_x f(x).
\label{eqn:selection_problem}
\end{align}
Under a mild technical condition, the next theorem produces a majorizer set $\Phi \subseteq \mathbf{R}^M$ for problem~\eqref{eqn:selection_problem} that is independent of the selection rules $\{S_m\}_{m=1}^M$, that is, a set $\Phi$ that only depends on $M$ atoms $\{f_m\}_{m=1}^M$  that make up the selector $f\in \mathcal{S}_M$. This is useful because, in general, the selection rules $S_m$ are intractable, non-convex regions in $\mathbf{R}^D$ while the atoms $f_m$ have simple algebraic expressions. In concrete, we typically have closed-form expressions for the derivatives of $f_m$ and can compute the non differentiable points of $f_m$. Majorizer $\Phi$ will use these computations on atoms $\{f_m\}_{m=1}^M$. Assuming it exists, let $\nabla f(x)$ denote the gradient of function $f$ at point $x$.
In simple terms, Theorem~\ref{thm:majorization_theorem} combines the selection structure of $f\in \mathcal{S}_M$  with the well-known first-order necessary condition for optimality; namely if $x^*$ is a minimizer of $f$ and $f$ is differentiable at point $x^*$ then $x^*$ must be a stationary point of $f$, that is, $\nabla f(x^*)=0$. Given an arbitrary set $\Phi \subseteq \mathbf{R}^M$  let $\text{cl } \Phi$ denote the closure\footnote{The closure of $\Phi\subseteq \mathbf{R}^M$ is defined as the set of limits points of $\Phi$ so $\text{cl }  \Phi=\{x:  \exists\, \{x_n\}_{n\geq 1}\subseteq \Phi, x_n\rightarrow_n x \}$.} of $\Phi$ and $\text{int } \Phi$ its interior\footnote{The interior of $\Phi\subseteq \mathbf{R}^M$ is defined as the set of interior points of $\Phi$ so $\text{int }  \Phi=\{x:  \exists\,\, \epsilon>0,\enspace B(x,\epsilon)\subseteq \Phi\}$ with $ B(x,\epsilon)$ an open ball in $\mathbf{R}^M$ so $B(x,\epsilon)=\{v: ||x-v||<\epsilon\}$.}. The boundary of $\Phi$, denoted as $\partial\, \Phi$, is the set of points that belong to the closure of $\Phi$ but not to its interior,
\begin{align}
\partial \, \Phi= \text{cl } \Phi \setminus \text{int } \Phi.
\label{eqn:boundary_definition}
\end{align}
\begin{theorem}
	\label{thm:majorization_theorem}
	Let $f\in \mathcal{S}_M$ denote an $M$-th selector with atoms  $\{f_m\}_{m=1}^M$ and selection rules $\{S_m\}_{m=1}^M$. Assume that atoms $f_m$ are constant along the boundaries of the partition $S_m$, that is, any point $x$  in the intersection of two distinct boundaries $\partial S_{m_1},\partial S_{m_2}$ (for $m_1\neq m_2$) achieves the same value
	\begin{align}
m_1\neq m_2:\enspace x\in \partial S_{m_1} \cap \partial S_{m_2}   \Rightarrow  f_{m_1}(x)=f_{m_2}(x).\enspace \enspace
\label{eqn:atoms_constant_along_boundaries}
	\end{align}
	 Let $x^*$ denote a minimizer of $f$. Then $x^*$ must be a stationary point of some atom $f_m$ or a non-differentiable point of yet another atom  $f_{\overline{m}}$, or there must exist two distinct atoms $m_1\neq m_2$ that coincide on $x$ so $f_{m_1}(x)=f_{m_2}(x)$. In compact notation 
	\begin{align} 
	\arg \min_x f(x)\neq \emptyset \enspace \Rightarrow \enspace
	\arg \min_x f(x)  \subseteq  & \enspace \Phi
	\label{eqn:majorization_Theorem}
	\end{align}
	with majorizer $\Phi\subseteq \mathbf{R}^D$ given by
	\begin{align}
	\Phi=&\bigcup_{m=1}^M  \{x: \nabla f_m(x)=0  \} \,\cup  \bigcup_{\overline{m}=1}^M  \{x: \nexists \,\, \nabla f_{\overline{m}}(x) \}  \\
	&\cup    \bigcup_{m_1=1}^{M-1}  \,\, \bigcup_{m_2=m_1+1}^{M}\{x: f_{m_1}(x)=f_{m_2}(x) \}.
	\label{eqn:majorizer}
	\end{align}
\end{theorem}
\begin{proof}
   Let $x^*$ be a minimizer of $f$. By a first-order criterion, either $x^*$ is a stationary point ($\nabla f(x^*)=0$) or $f$ is non-differentiable at $x^*$. In compact notation
	\begin{align}
		\arg \min_x f(x)  \subseteq  & \enspace \underbrace{\{x: \nabla f(x)=0 \}}_{ \textstyle :=\Phi_1} \cup\underbrace{\{x: \nexists \,\,  \nabla f(x)=0 \}}_{ \textstyle :=\Phi_2}.\enspace \enspace
		\label{eqn:majorizer_aux_1}
	\end{align}
	Note that~\eqref{eqn:majorizer_aux_1} is already a valid majorizer of problem~\eqref{eqn:selection_problem} (although not particularly useful due to the selection structure of $f$). The proof proceeds by successively upper bounding the right-hand side of~\eqref{eqn:majorizer_aux_1} until we get~\eqref{eqn:majorizer}.  We start by analysing $\Phi_1$ by using partition $S_m$, so
	\begin{align}
	\Phi_1 = \bigcup_{m=1}^M \Phi_1 \cap S_m &\subseteq \bigcup_{m=1}^M \Phi_1 \cap \text{cl } S_m  \\
	&=\Big( \bigcup_{m=1}^M \Phi_1 \cap \partial\, S_m \Big) \cup \Big( \bigcup_{m=1}^M \Phi_1 \cap \text{int } S_m \Big)  \\
	&\subseteq \Big( \bigcup_{m=1}^M  \partial\, S_m  \Big) \cup \Big( \bigcup_{m=1}^M \Phi_1 \cap \text{int } S_m  \Big).
	\label{eqn:majorizer_aux_2}
	\end{align}
	The first inequality bounds each individual set $S_m$ by its closure $\text{cl } S_m$; the second equality uses~\eqref{eqn:boundary_definition}. Now if $x$ belongs to $\Phi_1 \cap \text{int } S_m$ then $f(x)=f_m(x)$ by~\eqref{eqn:selection_function_2}; furthermore $\nabla f(x)=\nabla f_m(x)$ since $x$ is an interior point of $S_m$.	
	 Using this property we upper bound $\Phi_1 \cap \text{int } S_m$ by
	\begin{align}
	\Phi_1 \cap \Big( \text{int } S_m&=\{x: \nabla f(x)=0 \} \Big) \cap \Big( \text{int } S_m \Big)\\
	                           &=\{x: \nabla f_m(x)=0 \} \cap \text{int } S_m \\
	                           &\subseteq \{x: \nabla f_m(x)=0 \}.	
	 	\label{eqn:majorizer_aux_3}                         
	\end{align}
	Combining bounds~\eqref{eqn:majorizer_aux_2} with~\eqref{eqn:majorizer_aux_3} yields
		\begin{align}
	\Phi_1 \subseteq \Big(\bigcup_{m_1=1}^M  \partial\, S_{m_1} \Big)  \cup \Big( \bigcup_{m=1}^M \{x: \nabla f_m(x)=0 \}\Big).
	\label{eqn:majorizer_aux_4}
	\end{align}
	where the index set on $\partial\, S_{m_1}$ was purposely changed to prepare for the next step. Let us fix an index $m_1$. Using assumption~\eqref{eqn:atoms_constant_along_boundaries}, we upper bound the boundary $\partial\, S_{m_1}$ as follows
	 		\begin{align}
	  \partial\, S_{m_1} &= \partial\, S_{m_1} \, \cap \,  \partial\, (\mathbf{R}^D \setminus S_{m_1})\\
	  &= \partial\, S_{m_1} \, \cap \,  \partial\, (S_1\cup \dots \cup S_{m_1-1} \cup S_{m_1+1} \dots \cup S_M)\\
	  &\subseteq \bigcup_{m_2\neq m_1}  \partial\, S_{m_1} \, \cap \,  \partial\, S_{m_2} \\
	  &\subseteq  \bigcup_{m_2\neq m_1} \{x: f_{m_1}(x)=f_{m_2}(x)\}.
	 \label{eqn:majorizer_aux_5}
	 \end{align}
	 The first equality holds because the boundary is closed under complements, that is, $\partial \, A= \partial\, (\mathbf{R}^D\setminus A)$ for any set $A$; the second equality follows because $\{S_m\}_{m=1}^M$ is a partition of $\mathbf{R}^D$; the first inequality uses the generic property $\partial\, (A \cup B) \subseteq \partial\, A \cup \partial\, B$. The final inequality uses assumption~\eqref{eqn:atoms_constant_along_boundaries}. Combining results~\eqref{eqn:majorizer_aux_4} and~\eqref{eqn:majorizer_aux_5} leads to
	 		\begin{align}
	 \Phi_1 \subseteq \hspace{-0.15cm} \bigcup_{m_1\neq m_2} \{x: f_{m_1}(x)=f_{m_2}(x)\}\,  \cup  \bigcup_{m=1}^M \{x: \nabla f_m(x)=0 \}.
	 \label{eqn:majorizer_aux_6}
	 \end{align}
	 The previous analysis can be adapted for $\Phi_2$ because, for an interior point $x\in \text{int } S_m$,  gradient $\nabla f(x)$ exists if and only if $\nabla f_m(x)$ exists.
	 Using similar arguments we majorize $\Phi_2$ by
	 	 		\begin{align}
	 \Phi_2 \subseteq \hspace{-0.15cm} \bigcup_{m_1\neq m_2} \{x: f_{m_1}(x)=f_{m_2}(x)\} \, \cup \bigcup_{\overline{m}=1}^M \{x: \nexists \,\, \nabla f_{\overline{m}}(x) \}.
	 \label{eqn:majorizer_aux_7}
	 \end{align}
	 Majorizer $\Phi$, in~\eqref{eqn:majorizer}, is the union of bounds~\eqref{eqn:majorizer_aux_6} and~\eqref{eqn:majorizer_aux_7}.
\end{proof}
\begin{remark}
	\label{remark:necessity_of_assumption}
	Simple examples show that condition~\eqref{eqn:atoms_constant_along_boundaries} is actually necessary for majorization~\eqref{eqn:majorizer} to hold, that is, if the atoms $\{f_m\}_{m=1}^M$ of $f$ are not constant along the boundaries of the selection rules $\{S_m\}_{m=1}^M$ then we cannot generally conclude that any minimizer $x^*$ of~\eqref{eqn:selection_problem} (assuming it exists) is contained in $\Phi$. As a concrete example consider $D=1$ and the selector $f\in \mathcal{S}_2$ with atoms and partition
\begin{align}
\	f_1(x)&= \begin{cases} |x+2| &\mbox{if } x<-1,\enspace \enspace S_1=(-\infty,-1.5),  \\
	0 & \mbox{if } x\geq -1 \end{cases}\\ f_2(x)&= \begin{cases} x+1 &\mbox{if } x< -1/2 \\
	1/2 & \mbox{if } x\geq -1/2 \end{cases},\enspace \, S_2=[-1.5,+\infty).
	\end{align}
Both rules $S_1,S_2$ have the same boundary $\partial S_1=\partial S_2=\{-1.5\}$ but different atom values $f_1(-1.5)=0.5 \neq -0.5=f_2(-1.5)$. So condition~\eqref{eqn:atoms_constant_along_boundaries} does not hold and, in this example,
	$\Phi$ is not a valid majorizer of problem~\eqref{eqn:selection_problem}. In concrete we get
	\begin{align}
	\arg \min_x f(x)=\{-1.5\} \nsubseteq & \{-2\} \cup [-1,+\infty) = \Phi.
	\end{align}
\end{remark}
It turns out that the localization problem~\eqref{eqn:percentile_min} is, in fact, a selection problem of the form~\eqref{eqn:selection_problem} due to the percentile function $p_L$. In concrete the objective of~\eqref{eqn:percentile_min} is an $M$-th selector function since it can be decomposed as
\begin{align*}
p_{L} \{ |y_m- \left\|x-a_m\right\|   |   \}= f_1(x) \mathbf{1}_{\mathcal{S}_1}(x)+\dots +f_M(x) \mathbf{1}_{\mathcal{S}_M}(x)
\label{eqn:decomposing_percentile}
\end{align*}
for atoms $f_m(x):=|y_m- \left\|x-a_m\right\| |$ and selection rules 
{\small\begin{align}
S_1&=\{x: f(x)=f_1(x) \} \\
S_2&=\{x: f(x)=f_2(x),f_1(x)\neq f(x) \}\\
S_3&=\{x: f(x)=f_3(x), f_1(x)\neq f(x), f_2(x)\neq f(x) \} \\
\vdots \\
S_M&=\{x: f(x)=f_M(x), f_1(x)\neq f(x), \dots, f_{M-1}(x)\neq f(x) \},
\label{eqn:selection_rules}
\end{align}}
with $f(x):=p_{L} \{ |y_m- \left\|x-a_m\right\|   |   \}$. In words, the atoms of $p_{L} \{ |y_m- \left\|x-a_m\right\|   |   \}$ are the model deviations $ x \mapsto |y_m- \left\|x-a_m\right\| $. The selection rule $S_{m_1}$ identifies the set of points such that $p_{M}\{f_m(x)\}=f_{m_1}(x)$, while insuring that $\{S_m\}_{m=1}^M$ is a proper partition~\footnote{ The collection $\{S_m\}_{m=1}^M$ is pairwise disjoint due to condition $f_{\overline{m}}(x)\neq f(x)$ for $\bar{m}\neq m$ in~\eqref{eqn:selection_rules}. The sets $\{S_m\}_{m=1}^M$ span $\mathbf{R}^2$  because the $L$-th percentile $p_L$ of a vector $z$  is an element of $z$, as seen in~\eqref{eqn:sort_elements}.} of $\mathbf{R}^2$. 
\\~\\
Given the connection between problems~\eqref{eqn:percentile_min} and~\eqref{eqn:selection_problem} a natural idea is to use Theorem~\ref{thm:majorization_theorem} to majorize the solutions\footnote{Problem~\eqref{eqn:percentile_min} has a non-empty minimizer set because the functions $x\mapsto |y_m- \left\|x-a_m\right\| $ are all coercive regardless of $m$.} of problem~\eqref{eqn:percentile_min}. If majorization~\eqref{eqn:majorization_Theorem} holds then we can, without loss of generality, reduce the search space from $\mathbf{R}^2$ to  $\Phi$ so
\begin{align}
\min_x p_{L} \{ |y_m- \left\|x-a_m\right\|   |   \}=\min_{x\in \Phi} p_L \{ |y_m- \left\|x-a_m\right\|   |   \}.\enspace 
\label{eqn:reduced_search_space}
\end{align}
Result~\eqref{eqn:reduced_search_space} is useful because it suggests a direct method to solve the original problem: we can simply create a finite grid over $\Phi$ and choose the grid point that yields the lowest objective. Let us note, however, that in order to apply Theorem~\ref{thm:majorization_theorem} we must 
show that the atoms $\{f_m\}_{m=1}^M$ are constant along the boundaries of the selection rules $\{S_m\}_{m=1}^M$. In general verifying~\eqref{eqn:atoms_constant_along_boundaries} for an arbitrary selector $f\in \mathcal{S}_M$ might be non-trivial. It turns out, however, that for a broad class of percentiles objectives assumption~\eqref{eqn:atoms_constant_along_boundaries} becomes simple to verify; the next lemma shows that condition~\eqref{eqn:atoms_constant_along_boundaries} holds for percentile objectives composed with continuous mappings.
\begin{lemma}
	\label{lemma:percentiles_composed_continuous}
	Given $M$ continuous functions  $\{f_m\}_{m=1}^M$ let $\{S_m\}_{m=1}^M$ denote the partition of~\eqref{eqn:selection_rules}. Then, the percentile function $x\mapsto p_L\big(f_m(x)\big)$ is constant along the boundaries of $\{S_m\}_{m=1}^M$, that is, condition~\eqref{eqn:atoms_constant_along_boundaries} holds for $	f=p_{L}\{ f_m \}$.
\end{lemma}
\begin{proof}
\textcolor{black}{	To prove~\eqref{eqn:atoms_constant_along_boundaries} we majorize the boundary $\partial \, S_{m_1}$ by using simple topological properties of the boundary and closure operators. In concrete,
	\begin{align}
	\partial \, S_{m_1} &\subseteq \text{cl } \, S_{m_1}\\  &\subseteq \text{cl }  \Big( \big\{x: f(x)=f_{m_1}(x)\big\} \Big), \\
	&\begin{aligned}=\text{cl } \, \Big( \bigcup\limits_{\mathcal{O}:\,\, |\mathcal{O}|=L} \big\{x:  &f_{m}(x)\geq  f_{m_1}(x),\enspace m\in \mathcal{O}, \\[-2ex] 
	& f_{m_1}(x)\geq f_{\bar{m}}(x), \bar{m}\in \mathcal{O}^C  \big\}    \Big)
	\end{aligned} \\
		&\begin{aligned}=\bigcup\limits_{\mathcal{O}:\,\, |\mathcal{O}|=L} \text{cl } \, \Big( \big\{x:  &f_{m}(x)\geq  f_{m_1}(x),\enspace m\in \mathcal{O}, \\[-2ex] 
		& f_{m_1}(x)\geq f_{\bar{m}}(x), \bar{m}\in \mathcal{O}^C  \big\}    \Big)
	\end{aligned} \\
	&\begin{aligned}=\bigcup\limits_{\mathcal{O}:\,\, |\mathcal{O}|=L} \Big( \big\{x:  &f_{m}(x)\geq  f_{m_1}(x),\enspace m\in \mathcal{O}, \\[-2ex] 
		& f_{m_1}(x)\geq f_{\bar{m}}(x), \bar{m}\in \mathcal{O}^C  \big\}    \Big)
	\end{aligned}
	\label{eqn:boundary_aux_1}
	\end{align}
	with $\mathcal{O}$ the set of outliers. The first inequality uses definition~\eqref{eqn:boundary_definition}. The second inequality makes uses of~\eqref{eqn:selection_rules} and the generic property $\text{cl}\, (A\cap B)\subseteq \text{cl}\,A\,  \cap \, \text{cl}\,B$. The first equality uses the percentile definition in $f=p_{L}\{f_m\}$. 
For the second equality note that $\text{cl }\, A \cup B = \text{cl }\, A \, \cup \, \text{cl }\, B$ for any sets $A,B$. The continuity assumption of $\{f_m\}_{m=1}^M$ justifies the final equality. In conclusion, for any $x$ in $\partial\, S_{m_1}$ there exists $L$ outliers in $\mathcal{O}$ such that $f_{m_1}$ is greater than the atoms in $\mathcal{O}^C$ and smaller than the atoms in $\mathcal{O}$. This is equivalent to $f(x)=f_{m_1}(x)$. Applying the same reasoning for $x$ in  $\partial\, S_{m_2}$ leads to $f(x)=f_{m_2}(x)$. Hence, for $x$ in the intersection $\partial\,S_{m_1} \, \cap \,\partial\,S_{m_2}$, we get the desired equality $f_{m_1}(x)=f(x)=f_{m_2}(x)$.}
\end{proof}
\begin{example}
	\label{example:application_majorization}
	To illustrate the usefulness of Theorem~\ref{thm:majorization_theorem} and Lemma~\ref{lemma:percentiles_composed_continuous}, we present a first example where the majorizer set $\Phi$ actually reduces to a finite set. In this case the aforementioned grid method actually solves the problem, in the sense that it leads to a finite-time algorithm that returns a provable minimizer of~\eqref{eqn:selection_problem}. The example as follows: we want to approximate a data vector $d\in \mathbf{R}^M$ by a scalar $x\in \mathbf{R}$, given that $d$ has $L$ outliers entries. Using the motivation of section~\ref{sec:introduction} we formalize this problem as 
	\begin{align}
	\min_x f(x),
	\label{eqn:problem_specific}
	\end{align}
	with $f(x)=p_L\{ | x-d_m|\}$. 
	\textcolor{black}{Assume all entries of  $d$ are distinct (so $d_{m_1}\ne d_{m_2}$ for $m_1\neq m_2$).} Then, simple computations give 
	\begin{align}
\Phi=&\bigcup_{m=1}^M  \Big\{d_m \Big\} \, \cup \,    \bigcup_{m_1=1}^{M-1}  \,\, \bigcup_{m_2=m_1+1}^{M} \Big\{ {d}_{m_1,m_2} \Big\},
\label{eqn:simple_Example}
	\end{align}
	with ${d}_{m_1,m_2}:=(d_{m_1}+d_{m_2})/2$ denoting the average of points $d_{m_1}$ and $d_{m_2}$. So any minimizer of~\eqref{eqn:problem_specific} must be a original data point $s_m$ or the average ${d}_{m_1,m_2}$ of two points. The majorization inequality~\eqref{eqn:majorization_Theorem} suggest a naive algorithm for problem~\eqref{eqn:problem_specific}: we simply evaluate the objective over set $\Phi$ and select the one which yields the lowest objective, that is 
		\begin{align}
	\min_x f(x) =\min \Big\{   \min_m f(d_m),  \min_{m_1\neq m_2}  f({d}_{m_1,m_2}) \Big\}.
	\label{eqn:problem_specific_2}
	\end{align}
	Figure~\ref{figure:example_atomic_majorizer} plots a realization of problem~\eqref{eqn:problem_specific} with $M=3$.
\end{example}

					\begin{figure}[h]
	\begin{center}
		\includegraphics[width=6cm]{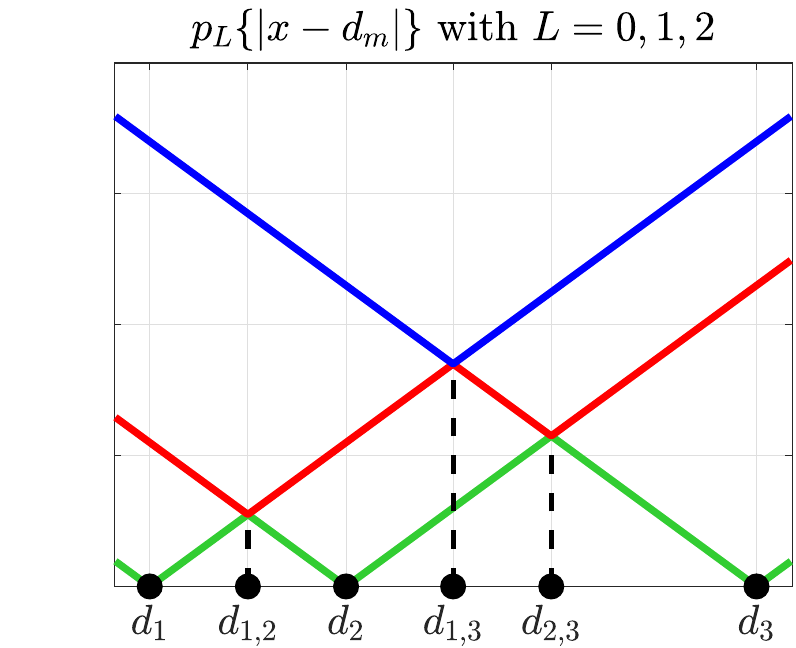}
		\caption{ Percentile objective $p_L\{ | x-d_m|\}$ with $M=3$ samples and number of outliers $L\in\{0,\dots,M-1\}$ varying: the blue, red and green curves represent the percentile function $p_L\{ | x-d_m|\}$ with $L=0,1,2$ respectively. The black points in the $x$- axis represent set $\Phi$ in~\eqref{eqn:simple_Example}. The blue function is convex because the mapping $p_{0}\{ | x-d_m|\}$ is equal to the maximum of convex functions $| x-d_m|$.	}
		\label{figure:example_atomic_majorizer}
	\end{center}
\end{figure}
\subsection{Parametrizing  Set $\Phi$}
\label{subsection: Parametrizing}
Theorem~\ref{thm:majorization_theorem} and Lemma~\ref{lemma:percentiles_composed_continuous} show that any minimizer $x^*$  of problem~\eqref{eqn:percentile_min} must belong to set $\Phi$ given by
\textcolor{black}{	\begin{align}
\Phi=& \bigcup_{m=1}^M  S_m \, \cup \,  \bigcup_{{m}=1}^M N_m  \,\, \cup   \hspace{-0.05cm}   \bigcup_{m_1<m_2} E_{m_1,m_2},
\label{eqn:majorizer_concrete_case_1}
\end{align}
with the auxiliary sets 
	\begin{align}
	S_m=&  \Big\{x:  \nexists \,\,  \nabla f_m(x)=0  \Big\},\\
	N_m=&  \Big\{x: \nexists \,\, \nabla f_{{m}}(x) \Big\}\,\\
	 E_{m_1,m_2}=&  \Big\{x:  f_{m_1}(x)=f_{m_2}(x)  \Big\},
	 \label{eqn:def_aux_sets}
\end{align}
defined in terms of atoms $f_m(x):=|y_m- \left\|x-a_m\right\|   |$.} This section derives an alternative representation of the set $\Phi$ which is sample friendly, that is, such that sampling points from $\Phi$ is computationally tractable. This enables a simple scheme to solve problem~\eqref{eqn:percentile_min} by creating a grid over $\Phi$ and simply choosing the best grid point. As seen in example~\ref{example:application_majorization}, this strategy delivers a provable minimizer of~\eqref{eqn:percentile_min} if $\Phi$ is a finite set. It turns out that set $\Phi$ given by~\eqref{eqn:majorizer_concrete_case_1} will be, in general, non-finite and even unbounded. Nonetheless, this grid approach often delivers very good estimates of the true target $x$, which tend to outperform several benchmarks in robust localization -- see section~\ref{sec:numerical_results}.
\\~\\
According to~\eqref{eqn:majorizer_concrete_case_1}, we first compute the set of stationary and non-differentiable points of atoms $f_m(x)=|y_m- \left\|x-a_m\right\|   |$. Function $f_m$ is differentiable if and only if $ x\neq a_m$ and $ \left\|x-a_m\right\| \neq y_m$. An application of the chain rule~\cite{magnus2019matrix} yields
\begin{align}
\nabla f_m(x)= -\frac{x-a_m}{\left\|x-a_m\right\|} \text{sign} (y_m- \left\|x-a_m\right\| )
\label{eqn:gradient}
\end{align}
with $\text{sign}(t)$ the sign of $t\neq 0$ ($ \text{sign}(t)=1$ if $t>0$ and  $ \text{sign}(t)=-1$ when $t<0$). Since the gradient $\nabla f_m(x)$ computed in~\eqref{eqn:gradient} is always non-zero for  $x\neq a_m$ and $\left\|x-a_m\right\| \neq y_m$ we conclude that there exists no stationary point of atom $f_m$, i.e., $ S_m=\emptyset$. Additionally, the set of non-differentiable points of atom $f_m$ is equal to
\begin{align}
  N_m=\{a_m\} \cup \{x:\left\|x-a_m\right\| = y_m\}.\enspace \enspace \enspace
  \label{eqn:computations_1}
\end{align}
Since $x$ lies in $\mathbf{R}^2$ we can easily parameterized sets $N_m$ by
\begin{align}
N_m= \{a_m\} \cup \Big\{ a_m+y_m(\cos \theta, \sin \theta): \theta\in [0,2\pi] \Big\}.
\label{eqn:circle_parametrization}
\end{align}
To parametrize the third term of~\eqref{eqn:majorizer_concrete_case_1} we fix a pair of distinct indices $m_1\neq m_2$ and analyze the set of points where the atoms $f_{m_1}$, $f_{m_2}$ yield the same value, that is,
\begin{align}
	E_{m_1,m_2}=\{x: |y_{m_1}- \left\|x-a_{m_1}\right\|   |=|y_{m_2}- \left\|x-a_{m_2} \right\|   | \}.
	\label{eqn:set_same_value}
\end{align}
 Let us assume, without loss of generality, that $y_{m_1}\geq y_{m_2}$. In order to parametrize set $E{m_1,m_2}$ we use two well know curves in plane geometry: standard ellipses and hyperbolas.
\begin{definition}(Pins-and-string Characterization of an Ellipse)
	Given constants $k,c\geq 0$ consider the  symmetric focal points $f_1:=(-c,0)$, $f_2:=(c,0)$. For $k>\left\|f_1-f_2\right\| =2c $ the standard ellipse is defined as the set of $x\in \mathbf{R}^2$ points whose additive distances to $f_1$ and to $f_2$ equals $k$, so
	\begin{align}
	\mathcal{E}(c,k):=\{x:  \left\|x-f_1\right\|+ \left\|x-f_2\right\| =k \}.
	\label{eqn:ellipse_definition}
	\end{align}
\end{definition}
\begin{definition}(Pins-and-string Characterization of an Hyperbola)
	Given constants $k,c\geq 0$ consider the  symmetric focal points $f_1:=(-c,0)$, $f_2:=(c,0)$. For $k<\left\|f_1-f_2\right\| =2c $ the standard hyperbola is defined as the set of $x\in \mathbf{R}^2$ points whose subtractive distances to $f_1$ and to $f_2$  equals $\pm k$, so
	\begin{align}
	\mathcal{H}(c,k):= \mathcal{H}_{+}(c,k) \cup \mathcal{H}_{-}(c,k),
	\label{eqn:def_hyperbola}
	\end{align}
	with the half-hyperbolas $ \mathcal{H}_{+}(c,k)$, $ \mathcal{H}_{-}(c,k)$ defined by
		\begin{align}
	\mathcal{H}_+(c,k)&:= \{x:  \left\|x-f_1\right\|- \left\|x-f_2\right\| =+k \}, 		\label{eqn:half_hyperbola_definition} \\
	\mathcal{H}_-(c,k)&:= \{x:  \left\|x-f_1\right\|- \left\|x-f_2\right\| =-k \}.
	\end{align}
\end{definition}
 It turns out that set $E{m_1,m_2}$ is closely related to standard ellipses $\mathcal{E}(c,k)$ and standard half hyperbolas $\mathcal{H}_+(c,k)$. In concrete, after carefully translating and rotating $E{m_1,m_2}$ we can get (1) a standard ellipse $\mathcal{E}(c,k)$, or (2) a standard half-hyperbola $	\mathcal{H}_+(c,k)$ or (3) their union $\mathcal{E}(c,k) \cup \mathcal{H}_+(c,k)$. The switching between these three regimes depends on the anchor-measurement data $(a_{m_1},a_{m_2},y_{m_1},y_{m_2})$. The next lemma designs a suitable translation and rotation that morph $E{m_1,m_2}$ as required. Given any vector $x$ in $\mathbf{R}^2\setminus\{0\}$ let $\arctantwo(x)\in(-\pi,\pi]$ denote the angle of vector $x$  in polar coordinates.
\begin{lemma}
	\label{lemma:rotated_translated_set}
	Given anchors $a_{m_1}$, $a_{m_2}\in \mathbf{R}^2$ and measurements $y_{m_1},y_{m_2}$ assume $y_{m_1}\geq y_{m_2}$. Then set $E_{m_1,m_2}$ is given by
 \begin{align*}
E_{m_1,m_2}=\overline{a}+R\begin{cases} \, \mathcal{E}(c,k_E) & \text{if } 2c <k_H \\
\mathcal{H}_{+}(c,k_H) &  \text{if } 2c>k_E \\
\mathcal{E}(c,k_E) \cup \mathcal{H}_{+}(c,k_H) & \text{if }k_H\leq 2c\leq k_E
\end{cases},
\label{eqn:complete_characterization}
	\end{align*}
	with $\overline{a}=(a_{m_1}+a_{m_2})/{2}$ the mean anchor position, $(c,k_E,k_H)=(\left\| a_{m_2}-\overline{a} \right\|,y_{m_1}+y_{m_2},y_{m_1}-y_{m_2})$ the constants that define the sets $\mathcal{E}(c,k_E) $ and $\mathcal{H}_{+}(c,k_H)$, and $R$ the rotation matrix associated with angle $\theta^*$ of vector $ a_{m_2}-\overline{a}$,
	\begin{align}
			R&=\begin{bmatrix}
	\cos \theta^* & -\sin \theta^* \\
	\sin \theta^* & \cos \theta^*
	\end{bmatrix},\enspace \theta^* := \arctantwo  ( a_{m_2}-\overline{a}).
	\end{align}
\end{lemma}
\begin{proof}
     We start by opening the absolute values in definition~\eqref{eqn:def_aux_sets}, so
     \begin{align}
     E_{m_1,m_2}= E_{m_1,m_2}^{S} \cup E_{m_1,m_2}^{O},
     \end{align} 
     with $$E_{m_1,m_2}^{O}=\{x: \left\|x-a_{m_1}\right\|+\left\|x-a_{m_2} \right\|   =y_{m_1}+y_{m_2}  \}$$ and $$E_{m_1,m_2}^{S}=\{x: \left\|x-a_{m_1}\right\|-\left\|x-a_{m_2} \right\|   =y_{m_1}-y_{m_2}  \}.$$ Set $E_{m_1,m_2}^{O}$ assumes that the absolute values in $E_{m_1,m_2}$ have opposite signs (hence the subscript O of opposite) while $E_{m_1,m_2}^{S}$ assumes equal signs (hence the subscript S of same). To proceed we consider three separate cases: $2c>k_E$, $2c<k_H$  and $k_H\leq 2c\leq k_E$. 
     For $2c>k_E$ there exists no target position $x$ which yields opposite signs on $E_{m_1,m_2}$, that is, $E_{m_1,m_2}^{O}=\emptyset$. Indeed assume there exists an $x$ with
     \begin{align}
     \left\|x-a_{m_1}\right\|+\left\|x-a_{m_2} \right\|   =y_{m_1}+y_{m_2}.
     \end{align}
     Since $2c=\left\|a_{m_1}-a_{m_2} \right\|$ we may bound this constant by
     \begin{align}
     2c=\left\|a_{m_1}-a_{m_2} \right\| &\leq  \left\|x-a_{m_1}\right\|+\left\|x-a_{m_2}\right\|\\
                                        &=y_{m_1}+y_{m_2}\\
                                        &=k_E.
                                        \label{eqn:contradiction_sign_1}
     \end{align}
     But~\eqref{eqn:contradiction_sign_1} contradicts  $2c>k_E$. Hence $E_{m_1,m_2}^{O}=\emptyset$ when  $2c>k_E$. A similar reasoning shows that the absolute values in $E_{m_1,m_2}$ must have opposite signs when $2c<k_H$, that is, $E_{m_1,m_2}^{S}=\emptyset$ for $2c<k_H$. Hence
      \begin{align}
    E_{m_1,m_2}=\begin{cases} \,  E_{m_1,m_2}^{O} & \text{if } 2c <k_H \\
   E_{m_1,m_2}^{S} &  \text{if } 2c>k_E \\
    E_{m_1,m_2}^{O} \cup E_{m_1,m_2}^{S} & \text{if }k_H\leq 2c\leq k_E
     \end{cases}.
     \label{eqn:complete_characterization_aux_1}
     \end{align}
     To finish the proof we show that sets $E_{m_1,m_2}^{O}$ and $E_{m_1,m_2}^{S}$ are equal to a standard ellipse $\mathcal{E}(c,k_E)$ and a standard hyperbola $\mathcal{H}^{+}(c,k_H)$ (respectively) after a translation and rotation. Using the translation vector $\overline{a}$ and rotation matrix $R$ leads to
     \begin{align}
     E_{m_1,m_2}^{S}&=\{x: \left\|x-a_{m_1}\right\|-\left\|x-a_{m_2} \right\|   =y_{m_1}-y_{m_2}  \}\\
     &\begin{aligned}=\overline{a}+ R\,\{\tilde{x}: &\left\|\tilde{x}-\tilde{a}_{m_1}\right\|-\left\|\tilde{x}-\tilde{a}_{m_2} \right\|   =y_{m_1}-y_{m_2}  \},
     \end{aligned}
     \label{eqn:complete_characterization_aux_2}
     \end{align}
     with the change of variable $\tilde{x}=R^T\,(x-\overline{a})$ for any point $x$. Result~\eqref{eqn:complete_characterization_aux_2} follows because $R$ defines an isometry, that is,
     \begin{align}
     \forall\,\phi:\enspace 
 \left\|x-\phi\right\|&= \left\|R^T\,(x-\overline{a})-R^T\,(\phi-\overline{a})\right\|= \left\|\tilde{x}-\tilde{\phi}\right\|. \nonumber
     \end{align}
     Simple computations show that the change of variable $x\mapsto R^T\,(x-\overline{a})$ maps the anchors $a_{m_1}$ and $a_{m_2}$ to the focal points $(-c,0)$ and $(c,0 )$ (this is shown in figure~\ref{figure:set_Phi_cases} since the transformation $x\mapsto Rx$ rotates the input $x$ counter-clockwise by $ \theta^*$ radians). Hence set~\eqref{eqn:complete_characterization_aux_2} is equal to $\overline{a}+R\,\mathcal{E}(c,k_E)$. Doing the same change of variables in $E_{m_1,m_2}^{O}$ leads to $E_{m_1,m_2}^{O}=\overline{a}+R\,\mathcal{H}_{+}(c,k_H)$ as desired.
\end{proof}
\begin{remark}
	The quantities  $(\overline{a},c,k_E,k_H,R)$ introduced in Lemma~\ref{lemma:rotated_translated_set} depend on the indices $m_1,m_2$ that define the set $E_{m_1,m_2}$. So, in full rigour, these quantities should also be indexed by indices $m_1,m_2$ in order to be fully identifiable. Since this makes the notation denser we choose not to display the dependency of $(\overline{a},c,k_E,k_H,R)$ on  $m_1,m_2$. So, from now on, the symbols $(\overline{a},c,k_E,k_H,R)$ are always associated with the set $E_{m_1,m_2}$ for some pair $(m_1,m_2)$ understood implicitly.
\end{remark}
\begin{figure}[h]
	\centering
	\begin{subfigure}[t]{0.23\textwidth}
		\centering
		\includegraphics[width=4cm,height=3.5cm]{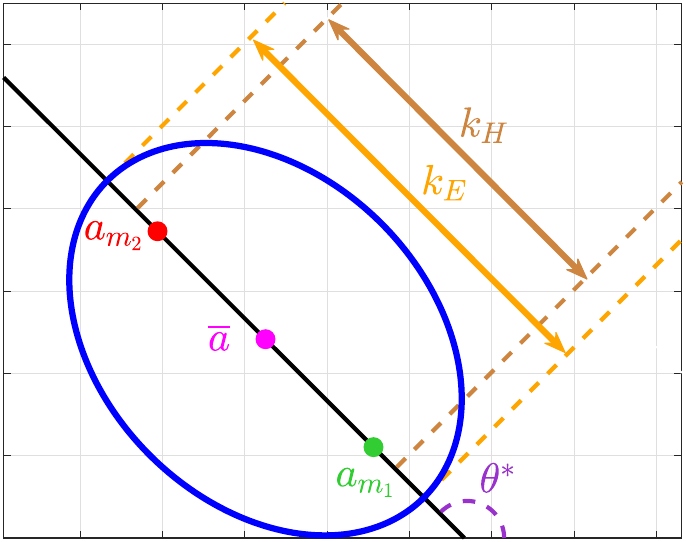}
		\caption{ Close-by anchors $2c<k_H$. }
				\label{figure:ellipse_set}
	\end{subfigure}%
	~ 
	\begin{subfigure}[t]{0.23\textwidth}
		\centering
		\includegraphics[width=4cm,height=3.5cm]{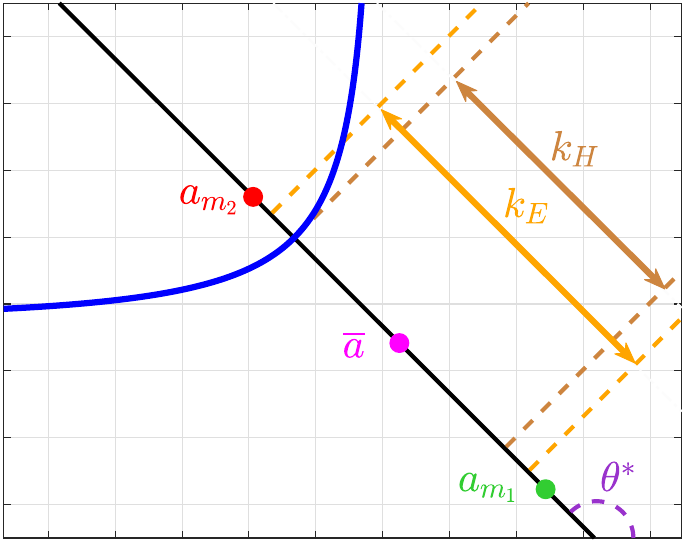}
		\caption{ Far-away anchors $2c>k_E$. }
				\label{figure:hyperbola_set}
		\vspace{0.3cm}
	\end{subfigure}
	\begin{subfigure}[t]{0.5\textwidth}
	\centering
	\includegraphics[width=5cm]{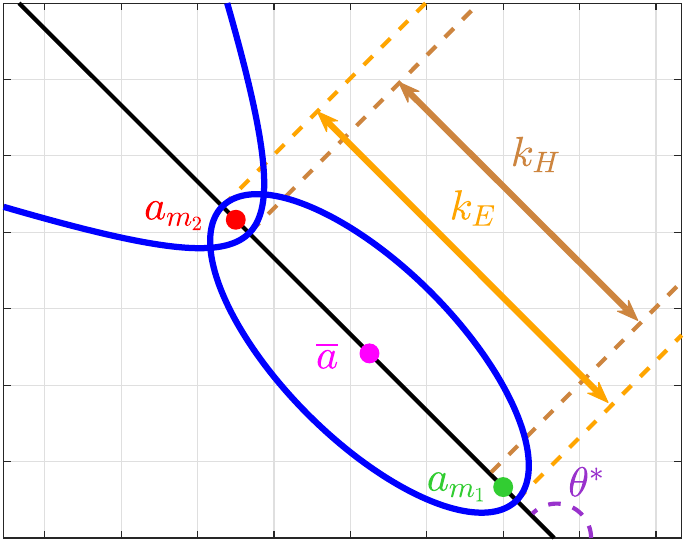}
	\caption{ Mid-range anchors $k_H\leq 2c\leq k_E$. }
		\label{figure:ellipse_plus_hyperbola_set}
\end{subfigure}
\caption{ Geometric interpretation of Lemma~\ref{lemma:rotated_translated_set}. The black line
	represents the set $\overline{a}+ (a_{m_2}-\overline{a})\, \mathbf{R}$; so a set centered at $\overline{a}$ and expanding linearly in direction $a_{m_2}-\overline{a}$; this set allows to plot the angle of direction $a_{m_2}-\overline{a}$, i.e., angle  $\theta^*$ (in purple). In all figures we plot set $E_{m_1,m_2}$ (in blue) with both constants $k_E$ (in orange) and $k_H$ (in brown) fixed. In figure~\ref{figure:set_Phi_cases} (c) the distance between anchors $a_{m_1}$ (in green) and $a_{m_2}$ (in red) insures that $k_H\leq 2c\leq k_E$. Figure~\ref{figure:set_Phi_cases} (a) achieves $2c<k_H$ by moving both anchors inwards towards the center $\overline{a}$; the anchors move outwards to satisfy $2c>k_E$ in figure~\ref{figure:set_Phi_cases} (b).}
	\label{figure:set_Phi_cases}
\end{figure}
Figure~\ref{figure:set_Phi_cases} plots the three generic forms of $E_{m_1,m_2}$. A consequence of Lemma~\ref{lemma:rotated_translated_set} is that the task of sampling from $E_{m_1,m_2}$ reduces to sampling from an ellipse~\eqref{eqn:ellipse_definition} or an half hyperbola~\eqref{eqn:half_hyperbola_definition}. Since these parametric curves are well known in plane geometry there exist simple parametrization of both sets. An ellipse $\mathcal{E}(c,k_E)$ can be parametrized by means of the trigonometric functions $(\cos \theta, \sin \theta)$, while an hyperbola $\mathcal{H}(c,k_H)$ uses the corresponding hyperbolic functions, that is
\begin{align}
\cosh t=\frac{e^t+e^{-t}}{2}, \enspace \enspace \sinh t=\frac{e^t-e^{-t}}{2}.
\label{eqn:hyperbolic_functions}
\end{align} 
The next lemma summarizes one standard parametrization of sets $\mathcal{E}(c,k_E)$ and $\mathcal{H}(c,k_H)$ by means of trigonometric and hyperbolic functions, respectively.
\begin{lemma}
	\label{lemma:trigonometric_hyperbolic}
	The standard ellipse $\mathcal{E}(c,k_E)$ and standard half hyperbolas $\mathcal{H}_+(c,k_H)$, $\mathcal{H}_-(c,k_H)$ can be parametrized by
	\begin{align}
	\mathcal{E}(c,k_E)&=  \Big\{ \big( \frac{k_E}{2}\cos \theta, \sqrt{\frac{k_E^2}{4}-c^2} \sin \theta\big): \theta \in [0,2\pi] \Big\}, \enspace \enspace \\
	\label{eqn:parametric_H_+}
	\mathcal{H}_+(c,k_H)&=  \Big\{ \big(\frac{k_H}{2}\cosh t, \sqrt{c^2-\frac{k_H^2}{4}} \sinh t\big): t \in \mathbf{R} \Big\},\\
		\mathcal{H}_-(c,k_H)&=  \Big\{  \big(-\frac{k_H}{2}\cosh t, \sqrt{c^2-\frac{k_H^2}{4}} \sinh t\big): t \in \mathbf{R} \Big\}.
	\end{align}
\end{lemma}
Lemmas~\ref{lemma:rotated_translated_set} and~\ref{lemma:trigonometric_hyperbolic} allow to efficiently sample set $E_{m_1,m_2}$ defined in~\eqref{eqn:set_same_value}. Combining this result with~\eqref{eqn:circle_parametrization} we express $\Phi$ in terms of computational tractable regions in $2$-D space: singletons, circles, ellipses and half-hyperbolas. The next section tackles the unboundedness of the half-hyperbolas.
\subsection{Dealing with the unboundedness of $\Phi$}
\label{subsection:unbounded_Phi}
As seen in Figure~\ref{figure:set_Phi_cases}, set $E_{m_1,m_2}$ is unbounded whenever $2c\geq k_H$ due to the half-hyperbolic component $\mathcal{H}_+(c_H,k)$. In this case we may sample $E_{m_1,m_2}$ by first discretizing the half hyperbola $\mathcal{H}_+(c,k_H)$ and then applying the affine transformation $x\mapsto \bar a + R\,x$. But discretizing a generic half hyperbola $\mathcal{H}_+(c,k_H)$ is computational intractable since the parameter $t$, in~\eqref{eqn:parametric_H_+}, can take any value in $\mathbf{R}$. Let us note, however, that this issue can be easily fixed since our half-hyperbola $\mathcal{H}_+(c,k_H)$ is not completely generic; it parametrizes set $\Phi$ that majorizes the minimizer set of problem~\eqref{eqn:percentile_min}. The objective function $x\mapsto p_{L} \{ | y_{{m}} -\left\|x-a_{m}\right\| | \}$ is clearly coercive since $ \left\|x\right\|\rightarrow \infty$ implies that  the deviation mapping $| y_{{m}} -\left\|x-a_{m}\right\| | $ grows unbounded for each $m$. So, clearly, a minimizer $x^*$ of~\eqref{eqn:percentile_min} cannot have a arbitrarily high norm $ \left\|x^*\right\| $. This intuitively prevents the parameter $t$ in $\mathcal{H}_+(c,k_H)$ to grow unbounded, since this would make both hyperbolic functions in~\eqref{eqn:hyperbolic_functions} also grow unbounded, that is
\begin{align*}
\text{min}\big\{|\cosh t|,|\sinh t|\big\}\rightarrow +\infty,
\end{align*} 
when $|t|\rightarrow +\infty$.
~\\
To formalize the previous argument we start  by designing a tractable family of upper bounds on the norm of a generic minimizer $x^*$ of~\eqref{eqn:percentile_min}.
\begin{lemma}
	\label{lemma:unbounded}
	Let $\hat{x}\in \mathbf{R}^2$ be arbitrary. Then, the solution set of~\eqref{eqn:percentile_min} is bounded since any minimizer $x^*$ of~\eqref{eqn:percentile_min} respects 
	\begin{align}
	\left\|x^*\right\| \leq  B(\hat{x}),
	\label{eqn:bound_norm_inequality}
		\end{align}
where
		\begin{align}
		B(\hat{x})=p_{L} \{ \big| y_{{m}} -\left\|\hat{x}-a_{m}\right\| \big| \} + \max_m \{   \left\|a_{m}\right\|+y_m   \}.\enspace 
		\label{eqn:bound_norm}
		\end{align}

\end{lemma}
\begin{proof}
	Assume there exists a minimizer $x^*$of~\eqref{eqn:percentile_min} with
	\begin{align}
		\left\|x^*\right\| > p_{L} \{ | y_{{m}} -\left\|\hat{x}-a_{m}\right\| | \} + \max_m \{   \left\|a_{m}\right\|+y_m   \}.\enspace 
	\label{eqn:bound_norm_inequality_aux_1}
	\end{align}
	To get a contradiction we first note that~\eqref{eqn:bound_norm_inequality_aux_1} is equivalent to
		\begin{align}
	\left\|x^*\right\|- \left\|a_{m}\right\|-y_m  > p_{L} \{ | y_{{m}} -\left\|\hat{x}-a_{m}\right\| | \},\enspace \forall \,m.
	\label{eqn:bound_norm_inequality_aux_2}
	\end{align}
	Using~\eqref{eqn:bound_norm_inequality_aux_2} we bound the objective $ p_{L} \{ | y_{{m}} -\left\|x^*-a_{m}\right\| | \}$ by
	\begin{align}
	p_{L} \{ \big| y_{{m}} -\left\|x^*-a_{m}\right\| \big| \}
	& \geq 	p_{L} \{\left\|x^*\right\|- \left\|a_{m}\right\|-y_m \}\\
	&>p_{L} \{ | y_{{m}} -\left\|\hat{x}-a_{m}\right\| | \}.
	\label{eqn:bound_norm_inequality_aux_3}
	\end{align}
	Result~\eqref{eqn:bound_norm_inequality_aux_3} says that $x^*$ yields a strictly larger objective than $\hat{x}$. This is impossible because $x^*$ is a minimizer of~\eqref{eqn:percentile_min}.
\end{proof}
Assume now that $E_{m_1,m_2}$ is unbounded because of the half-hyperbolic component $\mathcal{H}_+(c_H,k)$. Sampling 
$E_{m_1,m_2}$ is computationally intractable since the parameter $t$ in~\eqref{eqn:parametric_H_+} can take any value in $\mathbf{R}$. A natural idea to approximate $\mathcal{H}_+(c,k_H) $ is to bound the parameter $t$ to a finite interval, say $t\in [-U,U]$, and consider 
\begin{align}
\mathcal{H}^T_+(c,k_H)&=\Big\{ ( \frac{k_H}{2}\cosh t, \sqrt{c^2-\frac{k_H^2}{4}} \sinh t): |t|\leq U \Big\}.\enspace \enspace
\label{eqn:truncated_half_hyperbola}
\end{align}
The subscript $T$ in $\mathcal{H}^T_+(c,k_H)$ stands for truncated, since $\mathcal{H}^T_+(c,k_H)$ is a truncated version of $\mathcal{H}_+(c,k_H)$. Intuitively, high values of $U$ make  $\mathcal{H}^T_+(c_H,k)$ a good approximation of the half-hyperbola $\mathcal{H}_+(c_H,k)$. Hence, it is plausible to bound the minimizers of~\eqref{eqn:percentile_min} by considering a truncated version of $E_{m_1,m_2}$ where set $\mathcal{H}_+(c_H,k)$ is substituted by its truncated version $\mathcal{H}^T_+(c_H,k)$. Let $E_{m_1,m_2}^T$ denote such a set, that is
\begin{align}
E_{m_1,m_2}^T&=\bar{a}+R\,\begin{cases} \, \mathcal{E}(c,k_E) & \text{if } 2k <k_H. \\
\mathcal{H}^T_+(c,k_H) &  \text{if } 2k >k_E\\
\mathcal{E}(c,k_E) \cup \mathcal{H}^T_+(c,k_H) & \text{if } k_H\leq 2k\leq k_E.
\end{cases}
\label{eqn:majorizer_truncated_elements}
\end{align}
In general, the bound $U$ appearing in  $\mathcal{H}^T_+(c,k_H)$  will depend on the data $(a_{m_1},a_{m_2},y_{m_1},y_{m_2})$ defining set $E_{m_1,m_2}$. But, for ease of notation, we choose to not represent this dependency. Set $E_{m_1,m_2}^T$ is now computational tractable, in the sense that we can efficiently sample the truncated half hyperbola $ \mathcal{H}^T_+(c_H,k)$ for $t\in[-U,U]$. Given sets $E_{m_1,m_2}^T$ we mimic the structure~\eqref{eqn:majorizer_concrete_case_1} of $\Phi$ and consider the truncated set $\Phi^T\subseteq \mathbf{R}^2$
\begin{align}
\Phi^T=&\bigcup_{{m}=1}^M \{a_m\} \cup \big\{ a_m+y_m(\cos \theta, \sin \theta): \theta\in [0,2\pi] \big\}.  \\
&\cup    \bigcup_{m_1=1}^{M-1}  \,\, \bigcup_{m_2=m_1+1}^{M } E_{m_1,m_2}^T,
\label{eqn:majorizer_truncated}
\end{align}
where sets $E_{m_1,m_2}$ were substituted by its truncated version $E_{m_1,m_2}^T$. Set $\Phi^T$ is now sample friendly since each set $E_{m_1,m_2}^T$ is sample friendly.  The next lemma tunes the value of $U$ such that $\Phi^T$ a valid majorizer of problem~\eqref{eqn:percentile_min}. The main insight is to use Lemma~\ref{lemma:unbounded} to bound the norm of any point in $\mathcal{H}_+(c,k_H)$.
\begin{lemma}
	\label{lemma:cheap_bounds}
	 Given any $\hat{x}\in \mathbf{R}^2$ consider the bound
	\begin{align}
	U &=\log( \hat{U} + \sqrt{\hat{U}^2-1}    ),
		\label{eqn:upper_bounds}
		\\
	 \hat{U}&=\sqrt{\frac{\big( B(\hat{x})+  c+ ||a_{m_2}||\big)^2+ (c^2-{k_H^2}/{4})}{c^2}}
	 \label{eqn:upper_bounds_hat_U}
	\end{align}
	with $B(\hat{x})$ the norm bound defined in~\eqref{eqn:bound_norm}, $(c,k_H)$ the constants of $\mathcal{H}_+(c,k_H) $ and $m_2$ the index of measurement $y_{m_2}\leq y_{m_1}$ as in Lemma~\ref{lemma:rotated_translated_set}. Let $\Phi^T$ denote the truncated version of $\Phi$ with bound~\eqref{eqn:upper_bounds}. Then $\Phi^T$ majorizes~\eqref{eqn:percentile_min}, that is
	\begin{align} 
\arg \min_x  p_{L  } \{ \big| y_{{m}} -\left\|x-a_{m}\right\| \big| \} \subseteq  & \enspace \Phi^T.
\label{eqn:majorization_Theorem_truncated}
\end{align}
\begin{proof}
	Let $x^*$ denote a minimizer of~\eqref{eqn:percentile_min}. We show that if $x^*$ belongs to an hyperbolic set $ \overline{a} + R\, \mathcal{H}_{+}(c,k_H)$ then it also belongs to its truncated version $ \overline{a} + R\, \mathcal{H}^T_{+}(c,k_H)$. 
	Assume $x^*$ belongs to $ \overline{a} + R\, \mathcal{H}_{+}(c,k_H)$ for some constants $(c,k_H)$ such that $2c\geq k_H$. In this case Lemma~\ref{lemma:trigonometric_hyperbolic} yields the parametric form
	\begin{align}
	x^*= \overline{a} + R\, \phi^*,\enspace \enspace \phi^*:=\begin{bmatrix}
	\frac{k_H}{2}\cosh t^* \\
	\sqrt{c^2-\frac{k_H^2}{4}}\sinh t^*
	\end{bmatrix}
	\end{align}
	for some $t^*\in \mathbf{R}$.  Our goal is to bound the parameter $t^*$ to the interval $[-U,U]$ such that $x^*$ belongs to the truncated set  $ \overline{a} + R\, \mathcal{H}^T_{+}(c,k_H)$. For this we use Lemma~\ref{lemma:unbounded} which already bounds  $||x^*||$ by $B(\hat{x})$ with $\hat{x}$ arbitrary. This yields
	\begin{align}
	\left\|\phi^*\right\| = \left\|R\phi^*\right\| &\leq \left\|\overline{a}+R\phi^*\right\|+ \left\|\overline{a}\right\| \\
	&= \left\|x^*\right\|+ \left\|\overline{a}\right\| \\
		&\leq B(\hat{x})+ c+ ||a_{m_2}||.
			\label{eqn:upper_bounds_aux_1}
	\end{align}
	The final equality uses definition $c=||\overline{a}-a_{m_2}||$ from Lemma~\ref{lemma:rotated_translated_set}. Squaring both sides of~\eqref{eqn:upper_bounds_aux_1} leads to
		\begin{align}
	\frac{k_H^2}{4} \cosh^2(t^*) +	\big( c^2- \frac{k_H^2}{4} \big) \sinh^2(t^*) \leq \big( B(\hat{x})+  c+ ||a_{m_2}|| \big)^2.
	\label{eqn:upper_bounds_aux_2}
	\end{align}
	The identity $\sinh^2(t^*)=\cosh^2(t^*)-1$ implies that
			\begin{align}
	\cosh(t^*)  \leq \sqrt{\frac{\big( B(\hat{x})+  c+ ||a_{m_2}||\big)^2+ (c^2-{k_H^2}/{4})}{c^2}}.\enspace \enspace
	\label{eqn:upper_bounds_aux_3}
	\end{align}
	The right hand side of~\eqref{eqn:upper_bounds_aux_3} is  greater than or equal to one since $B(\hat{x})$ is non-negative and  $c\geq k_H/2$. The inverse of the hyperbolic cosine $\cosh$ is equal to $z\mapsto \log (z+ \sqrt{z^2-1})$ when $z\in[1,+\infty)$. The inverse is an increasing function hence
				\begin{align}
	t^*  \leq \log( \hat{U} + \sqrt{\hat{U}^2-1}    ),
	\label{eqn:upper_bounds_aux_4}
	\end{align}
	with $\hat{U}$ equal to the right-hand side of~\eqref{eqn:upper_bounds_aux_3}.
\end{proof}
\end{lemma}
The final method for problem~\eqref{eqn:percentile_min} is summarized in Algorithm~\ref{alg:RPTE}. As in example~\ref{example:application_majorization}, the idea is to sample set $\Phi$ and choose the best grid point according to the percentile objective.
  \begin{algorithm}[htbp]
	\caption{Robust Percentile Target Estimator - RPTE}\label{alg:RPTE}
	\begin{enumerate}
		\item Input Parameters: number of outliers $L$, position of $2$-D anchors $\{a_m\}_{m=1}^M \subseteq \mathbf{R}^2$, measurements $\{y_m\}_{m=1}^M$ and number of grid points $G\geq 2$ to discretize $\Phi^T\subseteq \mathbf{R}^2$.
		\vspace{0.3cm}
		\item Discretize the domains $[0,2\pi]$ and $[-U,U]$ by creating $G$ uniformly spaced points for each interval. In concrete
		\begin{align}
		\theta_g= 2\pi \frac{g}{ G-1},\enspace  \enspace t_g= -U+2U\frac{g}{ G-1},
		\end{align}
		with index $g=0,\dots,G-1$.
		\vspace{0.3cm}
		\item \label{alg:step3}  Use the grid points $(\theta_g,t_g)$  to discretize the circunferences, ellipses and half-hyperbolas in~\eqref{eqn:majorizer_truncated}.  This yields a discrete version of majorizer $\Phi^T$ denoted as $\Phi^{T}_D$.
		\vspace{0.2cm}
		\item \label{alg:step4} Choose the grid point $\hat{x}\in \Phi^T_D $ with the lowest objective
		\begin{align}
		\hat{x} \in \min_{x\in\Phi^T_D } p_{L} \{ | y_{{m}} -\left\|x-a_{m}\right\| | \}.
		\end{align}
	\end{enumerate}
\end{algorithm}
\section{Numerical Results}
\label{sec:numerical_results}
In this section we benchmark our method with four state-of-the art algorithms in target localization~\cite{Becks_paper_2008},~\cite{Target_Tracking_by_GD_2021},~\cite{Reweighted_LS_2018},~\cite{STRONG_2021}. The first two methods~\cite{Becks_paper_2008},~\cite{Target_Tracking_by_GD_2021} are based on non-robust formulations of the problem, that is, formulations that ignore the existence of outlier measurements. We include these methods to highlight the importance of robust approaches. The  algorithms in~\cite{Reweighted_LS_2018},~\cite{STRONG_2021} optimize robust formulations of the problem.
The benchmark estimates are denoted by $\hat{x}_{\text{SR-LS}}$~\cite{Becks_paper_2008}, $\hat{x}_{\text{GD}}$~\cite{Target_Tracking_by_GD_2021} , $\hat{x}_{\text{Huber}}$~\cite{STRONG_2021} and $\hat{x}_{\text{SR-IRLS}}$~\cite{Reweighted_LS_2018}. Our estimate, denoted as $\hat{x}_{\text{RPTE}}$, is computed via Algorithm~\ref{alg:RPTE} with just $G=20$ grid points.
\\~\\
\begin{figure*}[h!]
	\centering
	\begin{subfigure}[t]{0.32\textwidth}
		\centering
		\includegraphics[width=5cm]{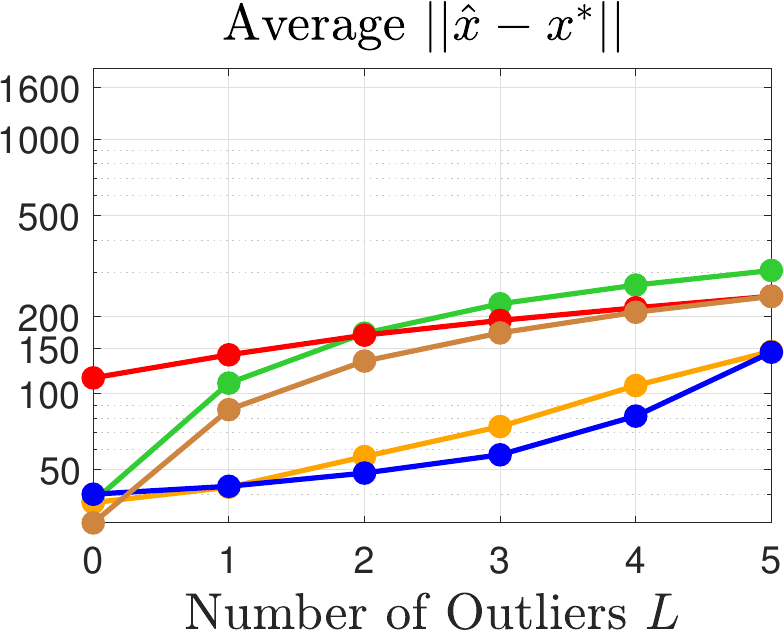}
		\caption{Outlier uncertainty $\sigma_{\mathcal{O}}=0.5$ [Km]. }
		\label{figure:result_1}
	\end{subfigure}%
	~ 
	\begin{subfigure}[t]{0.32\textwidth}
		\centering
		\includegraphics[width=5cm]{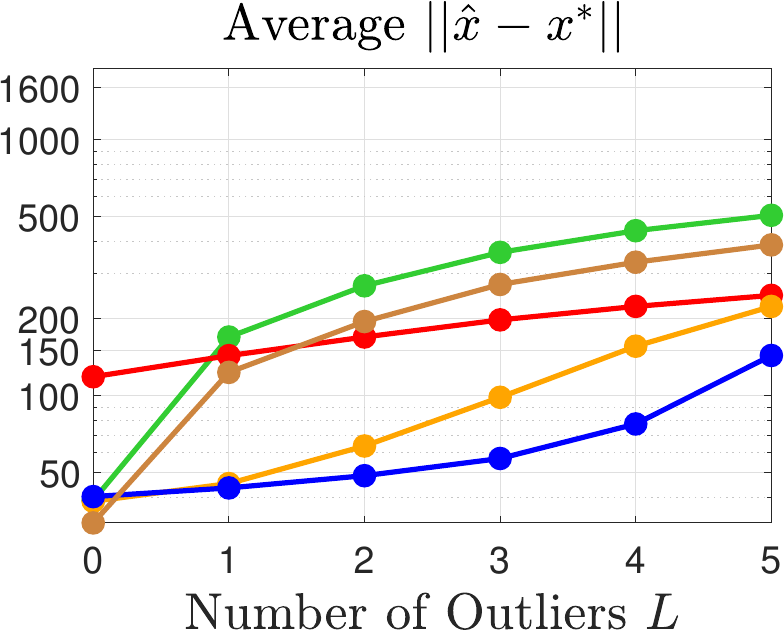}
		\caption{Outlier uncertainty $\sigma_{\mathcal{O}}=0.75$ [Km]. }
		\label{figure:result_2}
		\vspace{0.3cm}
	\end{subfigure}
	\begin{subfigure}[t]{0.32\textwidth}
		\centering
		\includegraphics[width=5cm]{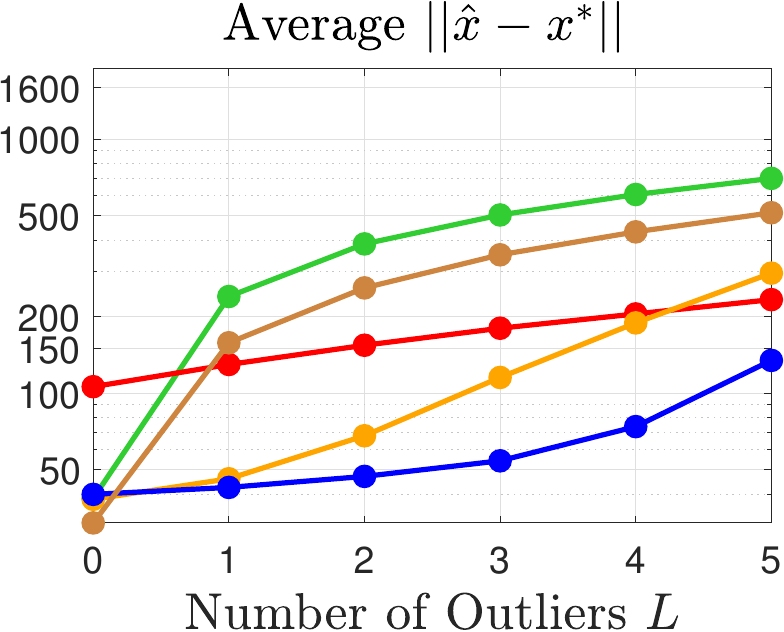}
		\caption{Outlier uncertainty $\sigma_{\mathcal{O}}=1$ [Km]. }
		\label{figure:result_3}
	\end{subfigure}
	~\\
	\begin{subfigure}[t]{1\textwidth}
		\centering
		\includegraphics[width=11cm]{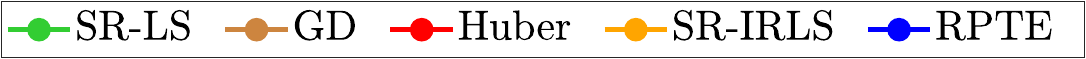}
	\end{subfigure}
	~\\
	\begin{subfigure}[t]{0.32\textwidth}
		\centering
		\includegraphics[width=5cm]{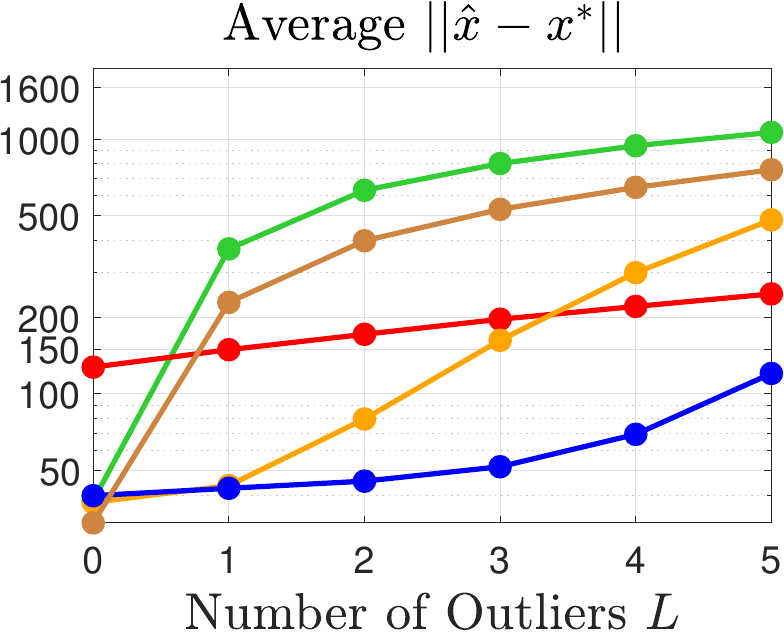}
		\caption{Outlier uncertainty $\sigma_{\mathcal{O}}=1.5$ [Km]. }
		\label{figure:result_4}
	\end{subfigure}
	\begin{subfigure}[t]{0.32\textwidth}
		\centering
		\includegraphics[width=5cm]{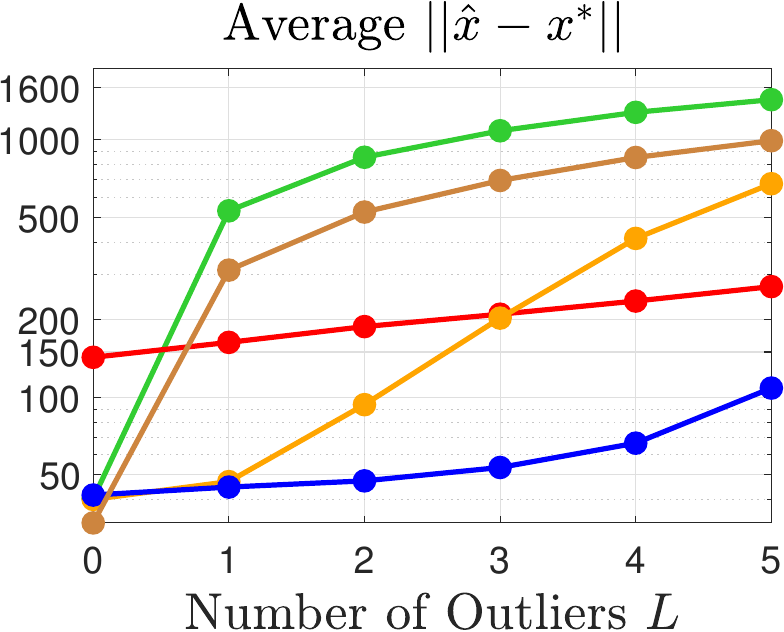}
		\caption{Outlier uncertainty $\sigma_{\mathcal{O}}=2$ [Km]. }
		\label{figure:result_5}
	\end{subfigure}
	\begin{subfigure}[t]{0.32\textwidth}
		\centering
		\includegraphics[width=5cm]{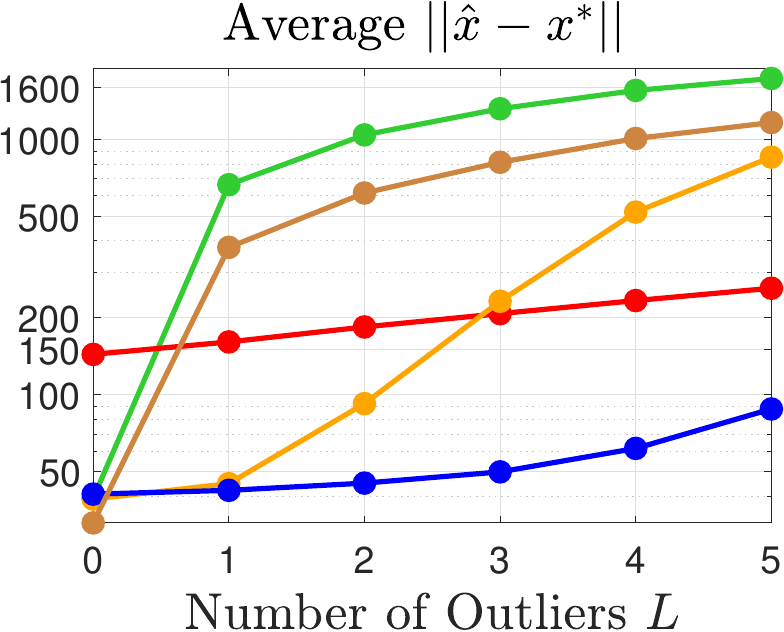}
		\caption{Outlier uncertainty $\sigma_{\mathcal{O}}=2.5$ [Km]. }
		\label{figure:result_6}
	\end{subfigure}
	\caption{ Average estimation error (in meters) for $5000$ Monte Carlo trials and increasing outlier uncertainties $\sigma_{\mathcal{O}}$.}
	\label{figure:numerical_results}
\end{figure*}
\noindent\textbf{Simulated setup.} As the localization setup, we consider a $1$ Km$^2$
square area and with $M=10$ anchors. The number of outliers, $L$ varies between $0$ (only inliers) and $M/2$ (half-outliers, half-inliers). As explained in Section~\ref{sec:introduction}, $y_m$ is an outlier if it introduces a large deviation in the model
\begin{align}
y_m=\left\|x^*-a_m\right\|+u_m,
\label{eqn:data_model_1}
\end{align}
with $x^*$ the true position of the target.
To distinguish inlier and outlier, we sample $u_m$ from a zero-mean Gaussian distribution with varying standard deviation: inlier measurements $\{y_m\}_{m\in \mathcal{I}}$ have a fixed standard deviation (uncertainty) of $50$ meters so $\sigma_{\mathcal{I}}=0.05$ Km; the uncertainty  $\sigma_{\mathcal{O}}$ of outliers varies in the kilometer grid  $\{0.5,0.75,1,1.5,2,2.5\}$.
For each value of $\sigma_{\mathcal{O}}$, we generate $5000$ problem instances as follows:
\begin{itemize}
	\item We generate $100$ positions for the anchors $\{a_m\}_{m=1}^M$ and target $x^*$  by sampling the unit square $[ 0, 1]^2$ uniformly.
	\item For each configuration of anchors $a_m$ and target $x^*$ we generate $50$ random sequences of outliers, i.e., $50$ lists \begin{align*}
		l^{(1)}&=(l_1^{(1)},\dots,l_{M/2}^{(1)}), \\
		&\vdots \\
		 l^{(50)}&=(l_1^{(50)},\dots,l_{M/2}^{(50)})\end{align*} of $M/2$ outlier indices. 
	 We generate $50$ lists instead of a single one such that, for a particular realization of anchors $a_m$ and target $x^*$, our experiments accommodate setups where the number of outliers $L$ varies (so taking the first $L$ indexes of each individual list) but also the outlier assignment varies for a fixed number of outliers (so taking the same index for different lists).  
	  For $L=0$, all measurements $y_m$ are generated by a Gaussian model 	with standard deviation  of $\sigma_{\mathcal{I}}$, that is
	\begin{align}
	y_m=|\left\|x^*-a_m\right\|+ \mathcal{N}(0,\sigma_{\mathcal{I}}^2)|.
	\label{eqn:measurement_generation_inliers}
	\end{align}
\textcolor{black}{	The absolute value in~\eqref{eqn:measurement_generation_inliers} simply insures that all measurements are non-negative.} For $L\geq 1$, we select the first $L$ indices of each list $l^{(1)},\dots,l^{(50)}$ to serve as the outliers of our model. Outlier measurements  have a standard deviation of $\sigma_{\mathcal{O}}$, that is
		\begin{align}
	y_m=|\left\|x^*-a_m\right\|+ \mathcal{N}(0,\sigma_{\mathcal{O}}^2)|.
	\label{eqn:measurement_generation_outlier}
	\end{align}
		\item For each configuration of anchors $\{a_m\}_{m=1}^M$ , true target $x^*$ and measurement vector $y$, we deploy all previously mentioned algorithms.
\end{itemize} 

\vspace*{0.2ex}
\begin{figure*}[h]
	\centering
	\begin{subfigure}[t]{0.32\textwidth}
		\centering
		\includegraphics[width=5.5cm]{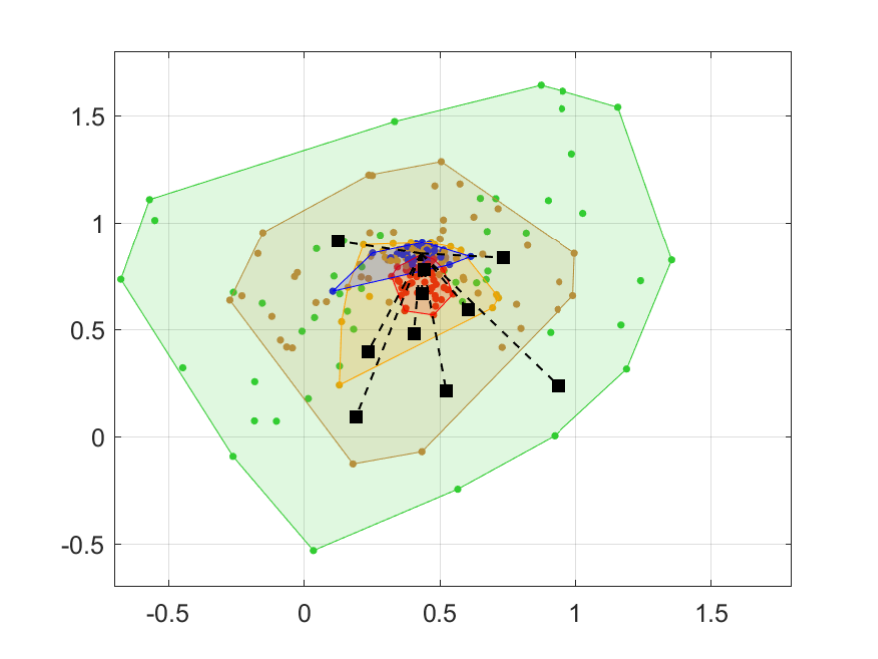}
		\caption{ $\sigma_{\mathcal{O}}=1$ [Km] and $L=3$ (zoom out). }
		\label{figure:realization_1}
	\end{subfigure}%
	~ 
	\begin{subfigure}[t]{0.32\textwidth}
		\centering
		\includegraphics[width=5.5cm]{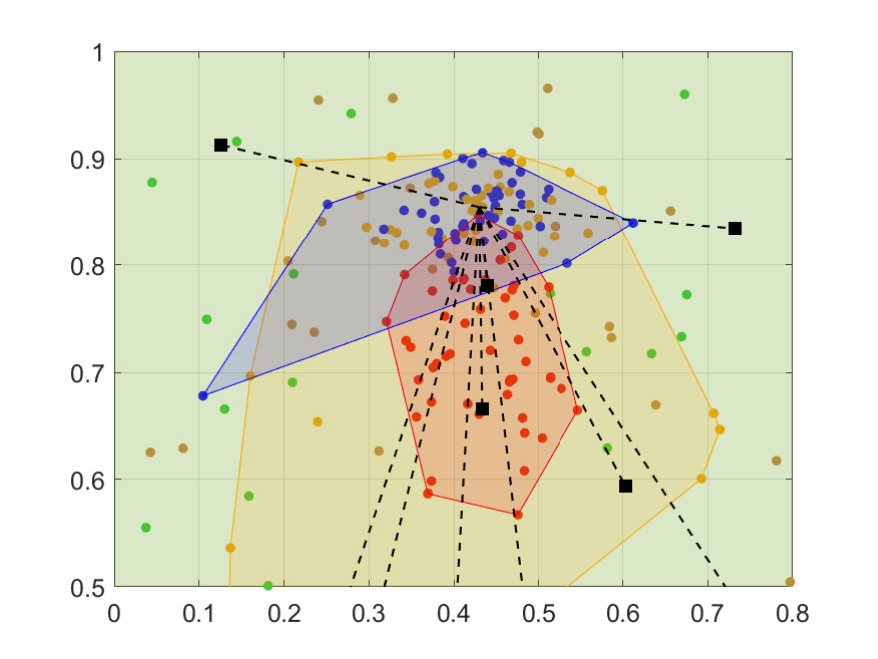}
		\caption{ $\sigma_{\mathcal{O}}=1$ [Km] and $L=3$ (zoom in). }
		\label{figure:realization_2}
		\vspace{0.3cm}s
	\end{subfigure}
	~ 
	\begin{subfigure}[t]{0.32\textwidth}
		\centering
		\includegraphics[width=5.5cm]{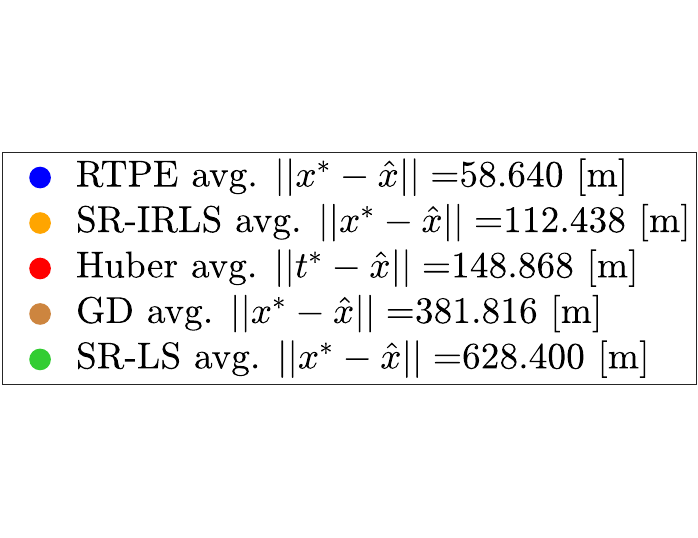}
		\caption{ $\sigma_{\mathcal{O}}=1$ [Km] and $L=3$ (errors). }
		\label{figure:realization_3}
		\vspace{-0.3cm}
	\end{subfigure}
	\begin{subfigure}[t]{0.32\textwidth}
		\centering
		\includegraphics[width=5.5cm]{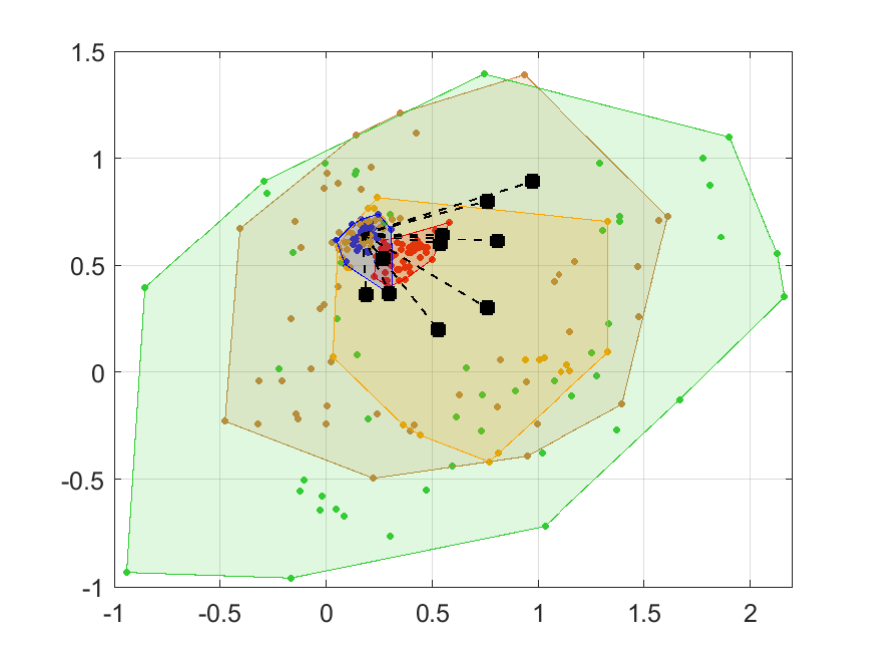}
		\caption{ $\sigma_{\mathcal{O}}=1.5$ [Km] and $L=4$ (zoom out). }
		\label{figure:realization_1_new}
	\end{subfigure}%
	~ 
	\begin{subfigure}[t]{0.32\textwidth}
		\centering
		\includegraphics[width=5.5cm]{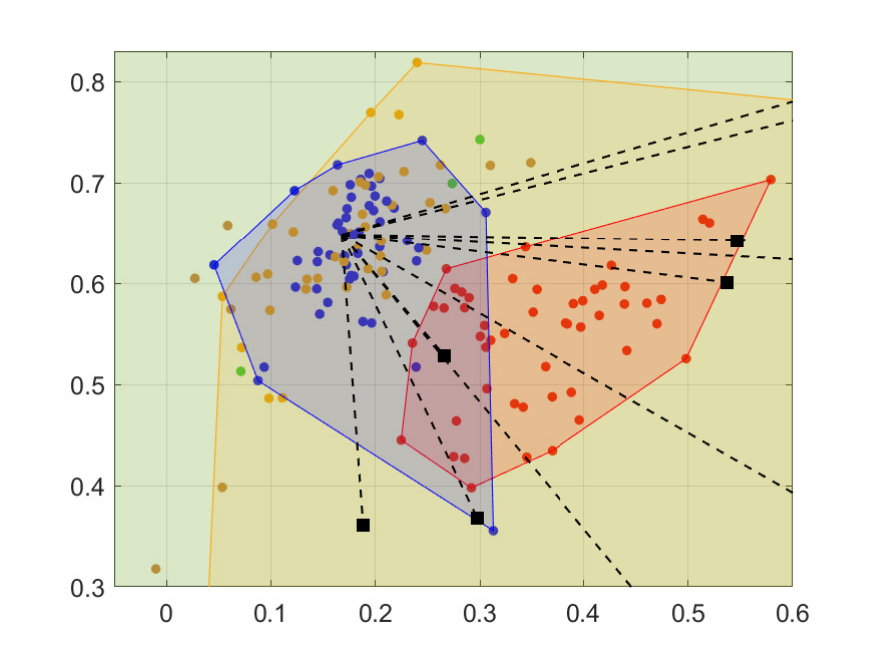}
		\caption{ $\sigma_{\mathcal{O}}=1.5$ [Km] and $L=4$ (zoom in). }
		\label{figure:realization_2_new}
		\vspace{0.3cm}
	\end{subfigure}
	~ 
	\begin{subfigure}[t]{0.32\textwidth}
		\centering
		\includegraphics[width=5.5cm]{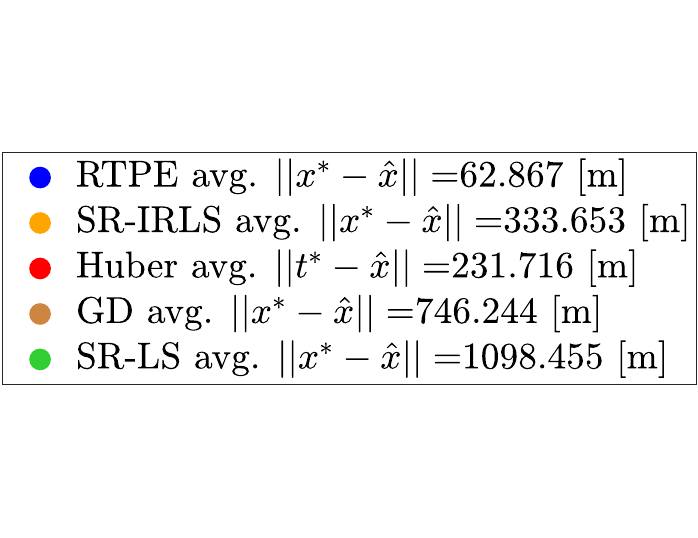}
		\caption{ $\sigma_{\mathcal{O}}=1.5$ [Km] and $L=4$ (errors). }
		\label{figure:realization_3_new}
		\vspace{0.3cm}
	\end{subfigure}
	\caption{ Estimates $\hat{x}$ for a subset of $50$ Monte Carlo trials with fixed anchors $a_m$ (black rectangles) and target $x^*$ (focal point of black dotted lines). The polyhedra represent the convex hull of estimates $\hat{x}$, using the color map of Figure~\ref{figure:numerical_results}.}
	\label{figure:epochs}
\end{figure*}
~\\
\noindent\textbf{Results.} 
Figure~\ref{figure:numerical_results} compares our percentile estimate $\hat{x}_{\text{Percentile}}$ with the benchmarks $\hat{x}_{\text{SR-LS}}$, $\hat{x}_{\text{GD}}$ , $\hat{x}_{\text{Huber}}$ and $\hat{x}_{\text{SR-IRLS}}$. For $L=0$ (no outliers) only the Huber method delivers a poor estimate $\hat{x}_{\text{Huber}}$ of target $x^*$. \textcolor{black}{In concrete, all other methods have an average reconstruction error lower than $\approx 40$ meters (in a $1$ Km$^2$ region, regardless of $\sigma_{\mathcal{I}}$) while the Huber estimate  $\hat{x}_{\text{Huber}}$ is $\approx 125$ meters away from $x^*$ (average error for different values of  $\sigma_{\mathcal{I}}$).} The upper bound of $40$ meters is met by the estimates $\hat{x}_{\text{Percentile}}$, $\hat{x}_{\text{SR-LS}}$ and  $\hat{x}_{\text{SR-IRLS}}$ while the gradient descent estimate $\hat{x}_{\text{GD}}$ tends to be $\approx 10$ meters closer to $x^*$. The higher performance of the gradient method is intuitive for two reasons: (1) \textcolor{black}{when $L=0$  the measurements $y_m$ are essentially modeled (up to the absolute in~\eqref{eqn:measurement_generation_inliers}) as i.i.d. samples of a Gaussian distribution with mean $\left\|x^*-a_m\right\|$ and standard deviation $\sigma_{\mathcal{I}}$;} (2) the gradient iterates are minimizing the range-based least squares objective (R-LS) which is proportional to the maximum likelihood objective~\cite{Becks_paper_2008}.  Pun \textit{et al.}~\cite{Target_Tracking_by_GD_2021} show that, under mild assumptions, there exists a neighbourhood where the R-LS objective is strongly convex. We suspect that the initialization proposed in~\cite{Target_Tracking_by_GD_2021} tends to be inside this neighbourhood. \textcolor{black}{In this case, Theorem 1 of~\cite{Target_Tracking_by_GD_2021} suggests that $\hat{x}_{\text{GD}}$ is an approximate maximum likelihood estimate of $x^*$, when no outlier exists.}
~\\
For $L\geq 1$ both non-robust estimates $\hat{x}_{\text{GD}}$, $\hat{x}_{\text{SR-LS}}$ exhibit a large performance decline. In concrete, a single outlier measurement ($L=1$) is capable of increasing the estimation error $ \left\|x^*-\hat{x}\right\|$ by alarming amounts. For example, when $\sigma_{\mathcal{O}}=1$ Km the estimation error $ \left\|x^*-\hat{x}\right\|$ tends to increase by a factor of $\approx 5$ (so the error $ \left\|x^*-\hat{x}_{\text{GD}}\right\| \approx 31$ m for $L=0$ becomes $ \left\|x^*-\hat{x}_{\text{GD}}\right\| \approx 158$ m for $L=1$). Furthermore the bias introduced in the estimates $\hat{x}_{\text{GD}}$, $\hat{x}_{\text{SR-LS}}$ increases with the outlier uncertainty $\sigma_{\mathcal{O}}$. For $\sigma_{\mathcal{O}}=2$ Km a single outlier now makes the estimation error $ \left\|x^*-\hat{x}_{\text{GD}}\right\|$ an order of magnitude (so $\approx 10$ times) worst. So the non-robust methods of~\cite{Becks_paper_2008},~\cite{Target_Tracking_by_GD_2021} fail to generalize in the presence of outliers.
\\~\\
For a larger number of outliers (say $L=2,3,4,5$) the most accurate estimates tends to be achieved by the proposed estimate $\hat{x}_{\text{RTPE }}$. Our method delivers significant accuracy gains with respect to the robust alternatives $\hat{x}_{\text{SR-IRLS}}$ and $\hat{x}_{\text{Huber}}$. For example when $\sigma_{\mathcal{O}}=1$ Km and $L=3$ the best robust benchmark $\hat{x}_{\text{SR-IRLS}}$ is, on average,  $115$ m away from $x^*$ while the percentile estimator $\hat{x}_{\text{RTPE}}$ doubles the estimation accuracy  ($ \left\|x^*-\hat{x}_{\text{RTPE}}\right\| \approx 54$ m) --- see the upper row of Figure~\ref{figure:epochs}. As before, higher values of $L$ and $\sigma_{\mathcal{O}}$ tend to accentuate the differences between methods.  For $\sigma_{\mathcal{O}}=1.5$ Km and $L=4$ both bencharmks yield similar estimates $\hat{x}_{\text{SR-IRLS}}$, $\hat{x}_{\text{Huber}}$  ($ \left\|x^*-\hat{x}_{\text{SR-IRLS}}\right\| \approx 300$ m,  $\left\|x^*-\hat{x}_{\text{Huber}}\right\| \approx 220$ m) yet the percentile method decreases the estimation error approximately $3$ to $4$ times  ($ \left\|x^*-\hat{x}_{\text{RTPE}}\right\| \approx 70$ m) --- see the lower row of Figure~\ref{figure:epochs}.  These findings suggest that the VaR methodology allows for robust localization in the presence of outlier measurements.
\\~\\
The higher performance of RTPE has an associated computational cost coming from the grid methodology of Algorithm~\ref{alg:RPTE}. Let us note, however, than in our experiments this cost is negligible. Indeed, on average, RTPE only takes $0.017$ seconds to compute $\hat{x}_{\text{RTPE}}$. This computational time is comparable to that of the benchmarks since  the fastest method takes $0.0009$ seconds to compute $\hat{x}_{\text{SR-LS}}$ while the  slowest approach takes $0.036$ seconds  to estimate $\hat{x}_{\text{Huber}}$. The low computational cost of RTPE comes from (a) considering a fine grid of $G=20$ in algorithm~\ref{alg:RPTE} and (b) having a reasonable number of anchors $M$  (say $M\leq 10$) which is typical in localization.
\section{Conclusion}
\label{sec:conclusion}
This paper addresses the problem of locating a target from range measurements, some of which may be outliers. Assuming that the number of outliers is known, we formulate a robust estimation problem using the percentile objective. This new formulation is given a statistical interpretation by considering the framework of risk analysis. In concrete, our formulation is equivalent to optimizing the value-at-risk (VaR) risk measure from portfolio optimization. To actually address the percentile problem we design a majorizer set which contains all minimizers of the original VaR formulation. Our majorizer can be parametrized by well known curves in plane geometry: singletons, circumferences, ellipses and bounded half-hyperbolas. Since all these regions have efficient parametrization we propose a grid algorithm -- RPTE -- which reduces the optimization problem to an efficient sampling scheme. Numerical experiments show that, on average, our method outperforms several benchmarks in target localization.

\bibliographystyle{IEEEtran}

\bibliography{IEEEabrv,M335}
\end{document}